\documentclass[11pt]{article}

\usepackage{amsfonts,amsmath,amssymb,amsthm}
\usepackage{indentfirst,color}
\usepackage[a4paper,ignoreall]{geometry}
\usepackage{hyperref,enumerate}
\usepackage{authblk}

\usepackage[T1]{fontenc}
\usepackage[utf8]{inputenc}

\newcommand{\od}{ {\rm ord}_2}
\newcommand{\tr}{ \rm Tr}

\newcommand{\gf}{ \mathbb{F}}

\newcommand{\ef}{ \mathbb{F}}

\newcommand{\ec}{ \mathbb{C}}

\newtheorem{theorem}{Theorem}
\newtheorem{corollary}{Corollary}
\newtheorem{proposition}{Proposition}
\newtheorem{problem}{Problem}
\newtheorem{lemma}{Lemma}
\newtheorem{definition}{Definition}
\newtheorem{example}{Example}

\newtheorem{remark}{Remark}

\def\yin#1 {\fbox {\footnote {\ }}\ \footnotetext { From Yin: {\color{blue}#1}}}
\def\hyin#1 {}

\makeatletter
\newcommand{\figcaption}{\def\@captype{figure}\caption}
\newcommand{\tabcaption}{\def\@captype{table}\caption}
\makeatother

\renewcommand\Authands{ and }

\begin{document}

\title{Quadratic Zero-Difference Balanced Functions, APN Functions and Strongly Regular Graphs}

\author[1]{Claude Carlet\thanks{Email: claude.carlet@univ-paris8.fr}}
\author[2]{Guang Gong\thanks{Email: ggong@uwaterloo.ca}}
\author[2]{Yin Tan\thanks{Corresponding author. Email: y24tan@uwaterloo.ca. Tel: +1 (519) 888-4567 EXT 32140.}}
\affil[1]{LAGA, Universities of Paris 8 and Paris 13, CNRS
\authorcr
	Department of Mathematics,  University of Paris 8  
2 rue de la libert\'e, 93526 Saint-Denis cedex 02, France }
\affil[2]{Department of Electrical and Computer Engineering\authorcr
        University of Waterloo, Ontario, Canada}

\renewcommand\Authands{ and }

\maketitle

\begin{abstract}
Let $F$ be a function from $\gf_{p^n}$ to itself and $\delta$ a positive integer. 
$F$ is called zero-difference $\delta$-balanced if the equation $F(x+a)-F(x)=0$ 
has exactly $\delta$ solutions for all nonzero $a\in\gf_{p^n}$. As a particular case, 
all known quadratic planar functions are zero-difference 1-balanced; and some quadratic 
APN functions over $\gf_{2^n}$ are zero-difference 2-balanced. In this paper, we study the 
relationship between this notion and differential uniformity; we show that all quadratic 
zero-difference $\delta$-balanced functions are differentially $\delta$-uniform and we 
investigate in particular such functions with the form $F=G(x^d)$, where $\gcd(d,p^n-1)=\delta +1$ 
and where the restriction of $G$ to the set of all nonzero $(\delta +1)$-th powers in $\gf_{p^n}$ 
is an injection. We introduce new families of zero-difference $p^t$-balanced functions. 
More interestingly, we show that the image set of such functions is a regular partial difference set, 
and hence yields strongly regular graphs; this generalizes the constructions of strongly regular 
graphs using planar functions by Weng et al. Using recently discovered quadratic APN functions on 
$\gf_{2^8}$, we obtain $15$ new $(256, 85, 24, 30)$ negative Latin square type strongly regular graphs. 


\end{abstract}

\section{Introduction}
\label{intro}

Functions defined over $\gf_{p^n}$ with high nonlinearity have been studied extensively in the last three decades
as they are widely used in symmetric cipher design, allowing resisting known attacks.
For instance, permutations over $\gf_{2^n}$ with high nonlinearity and
low differential uniformity (defined in Section 2.1) are
chosen as the Substitution boxes in block ciphers to bring the necessary confusion, as shown in \cite{nyberg}.
Besides this, they are interesting thanks to their close relationship with notions in coding theory and  
combinatorics. For instance, almost bent functions can be used to
construct codes, association schemes, graphs, authentication schemes \cite{CDN,codeandgraph} and classes of sequences with low cross-correlation \cite{HK}; bent functions can be used to
construct association schemes, strongly regular graphs \cite{tanas,tangraph} and classes of sequences with low cross-correlation as well \cite{HK};
quadratic almost perfect nonlinear functions (APN, defined in Section 2.1) can be used to construct dual hyperplanes \cite{edel};
perfect nonlinear functions (PN, defined in Section 2.1) can be used to construct difference sets
and strongly regular graphs \cite{weng}. For a survey of highly nonlinear
functions, one may refer to \cite{Cbook1, Cbook, carletandding}.
In this paper, we will study some quadratic functions having some peculiarity implying the 
properties above, and establish a new relationship between them and strongly regular graphs.

In \cite{ZDB-Ding}, Ding introduced a special kind of functions called \textit{zero-difference balanced}
(ZDB) functions to construct optimal constant composition codes. 
A function from an Abelian group $A$ to the other Abelian group $B$ is called zero-difference $\delta$-balanced
if the equation $F(x+a)-F(x)=0$ has exactly $\delta$ solutions for all nonzero $a\in A$, 
where $\delta$ is some positive integer. 
Recently, several classes of zero-difference balanced functions were constructed and more of their applications 
to the construction of combinatorial objects were explored, see
\cite{Han-ZDB,ZDB-Ding,Ding-DSS,Ding-ZDB,Ding-Tan,QiZhou-ZDB,ZTWY} and the references therein.
It is worth mentioning that the known constructions of ZDB functions are mostly defined on a cyclic group 
because such ZDB functions have more applications.
Throughout this paper, we shall consider ZDB functions on a non-cyclic group (more precisely the additive group 
of the field $\gf_{p^n}$) and show that some of them may be applied to construct partial difference sets.
We shall first show that
a quadratic zero-difference $\delta$-balanced function is a
quadratic differentially $\delta$-uniform function.
The converse of the above statement is not true in general:
by \cite{gohar-pott, weng}, all quadratic PN functions
are zero-difference 1-balanced, up to the addition of a linear function, 
but this is not true for APN functions on $\gf_{2^n}$. For example, $18$ of the $2,275$ newly discovered quadratic APN
functions on $\gf_{2^8}$ in \cite{yu-APN, tan-APN} are zero-difference 2-balanced; and when $n$ is odd, 
APN permutations are clearly not zero-difference 2-balanced. 
However, the notion of zero-difference balance gives a nice enlightenment on the APNness of 
some classes of known APN functions (see below).

For the construction of zero-difference $\delta$-balanced functions, it is shown in Corollary \ref{dto1map} that quadratic functions of the
form $F(x)=G(x^d)$ satisfy the requirement when $\gcd(d, p^n-1)=\delta +1$ and when
the restriction of $G$ to the set of $(\delta+1)$-th powers of $\gf_{p^n}$ is an injection.
By discovering such $G$, we obtain new families of differentially $\delta$-uniform functions.
In Section 4, on the one hand, we provide new methods to construct zero-difference $\delta$-balanced functions; and on the other hand, new families of zero-difference $p$-balanced functions
are presented. As a particular case, new APN functions are obtained.

It is proven in \cite{weng} that, given an even PN function on $\gf_{p^n}$ (i.e. $f(0)=0$ and $f(-x)=f(x)$ for all $x\in\gf_{p^n}\setminus\{0\}$),
its image set (excluding $0$) is either a Payley difference set when
$p^n\equiv 3\mod 4$, or a Payley partial difference set when $p^n\equiv 1\mod 4$. In Section 5,
we establish similar results. Precisely, let $F$ be a zero-difference $p^t$-balanced
function on $\gf_{p^n}$, where $n\equiv 0\mod 2t$ and $t>0$ (the case $t=0$ is actually the result obtained in \cite{weng}),
denoting $D=\mbox{Im}(F) \setminus \{0\}$, then $D$ is a regular
\begin{equation}
\label{srgpara}
    \left(p^n, \frac{p^n-1}{p^t+1},
               \frac{p^n-3p^t-2-\epsilon p^{n/2+2t}+\epsilon p^{n/2+t}}{(p^t+1)^2},
               \frac{p^n-\epsilon p^{n/2}+\epsilon p^{n/2+t}-p^t}{(p^t+1)^2}
    \right)
\end{equation}
partial difference set (PDS, defined in Section 2.3), where $n=2kt$ and $\epsilon=(-1)^{k}$.
Therefore, we obtain a new construction of strongly regular graphs (SRG, defined in Section 2.3)
by its relationship to partial difference sets.
Particularly, when $k$ is even, we obtain negative Latin square type SRGs.
In Section 6, using newly discovered zero-difference $2$-balanced (namely APN) functions
on $\gf_{2^8}$ and by comparing the SRGs with parameter $(256,85,24,30)$ to known constructions,
we found $15$ new such graphs.

The rest of the paper is organized as follows. In Section 2, we give necessary definitions and results.
The properties of zero-difference balanced functions are presented in Section 3.
Section 4 presents constructions of zero-difference $p^t$-balanced functions.
In Section 5, we establish the relationship between zero-difference $p^t$-balanced
functions and partial difference sets (strongly regular graphs), and in Section 6 we discuss the
newness of the SRGs obtained from zero-difference $2$-balanced functions.
Some concluding remarks are given in Section 7.

\section{Preliminary}
\label{pre}

In this section, we introduce basic definitions and results which will be used in the following sections.

\subsection{Functions defined over $\gf_{p^n}$}
\label{nonlinear-functions}

Let $F$ be a function from $\ef_{p^n}$ to itself.  For any $a,b\in\gf_{p^n}$; $a\neq 0$,
define the difference function $\delta_F(a,b) = |\{ x : x \in \gf_{p^n} | F(x+a) - F(x) = b \}|$,
where $|S|$ denotes the size of a set $S$. Let $\Delta_F = \max_{a,b\in\gf_{p^n}, a\neq 0} \delta_F(a,b)$,
the function $F$ is called a \textit{differentially $\Delta_F$-uniform function}.
Particularly, when $p=2$, it is easy to see that the smallest value of $\Delta_F$ is $2$,
we call a function with such value of $\Delta_F$ \textit{almost perfect nonlinear} (APN);
and when $p$ is odd and $\Delta_F=1$, we call such functions \textit{perfect nonlinear} (PN) or \textit{planar}.
The multiset $ \mathcal{D}_F:=\{ \delta_F(a,b): a \in \gf_{p^n}, b\in\gf_{p^n}, a \neq 0 \}$ is
called the \textit{differential spectrum} of $F$.

Another commonly used parameter evaluating the nonlinearity of $F$ is as follows.
For the above function $F$, the \textit{Walsh transform} of $F$ is defined as
\begin{equation*}
	\mathcal{W}_F(a,b) := \sum_{x\in\ef_{p^n}}\zeta_p^{\tr(aF(x)+bx)}, \ \ a\in\ef_{p^n}^*=\ef_{p^n}\setminus\{0\}, b\in\ef_{p^n},
\end{equation*}
where $\zeta_p$ is a complex primitive $p$-th root of unity, and $\tr(x):=\sum_{i=0}^{n-1}x^{p^i}$
denotes the usual trace function from $\ef_{p^n}$ to $\ef_p$.
The multiset $\mathcal{W}_F:=\{\mathcal{W}_F(a,b): a\in\ef_{p^n}^*, b\in\ef_{p^n} \}$ is called the
\textit{Walsh spectrum} of $F$, and each number $\mathcal{W}_F(a,b)$ is called the
\textit{Walsh coefficient} at $(a,b)$.  Particularly, for a $p$-ary function $f:\ef_{p^n}\rightarrow\ef_p$, the Walsh transform of $f$
equals $\sum_{x\in\ef_{p^n}}\zeta_p^{f(x)+\tr(bx)}$, which is denoted by $\mathcal{W}_f(b)$. When $p=2$, the \textit{nonlinearity} of $F$ is defined as $\mbox{NL}(F)=2^{n-1}-\frac{1}{2}\max_{a,b\in\gf_{2^n}, a\neq 0}\mathcal{W}_F(a,b)$.

Finally, the  \textit{algebraic degree} of a function $ F(x)=\sum_{i=0}^{p^n-1} a_ix^i\in\gf_{p^n}[x] $, denoted by $\deg F$,  is defined as
the maximal $p$-weight of the exponent $i$ such that $a_i\neq 0$,
where  {the $p$-weight of an integer $i$} is the sum in ${\Bbb Z}$ of the coefficients in
its $p$-ary expression. Particularly, $F$ is called \textit{quadratic} if $\deg F=2$
(sometimes it is called \textit{Dembowski-Ostrom}, in brief, DO, if $F(x)=\sum_{i,j=0}^{n-1} a_{ij} x^{p^i+p^j}$);
and $F$ is called \textit{affine} if $\deg F\leq 1$.
Two functions $F_1, F_2 :\gf_{p^n}\rightarrow\gf_{p^n}$ are called \textit{extended affine} (EA)-
equivalent if there exist linear permutations $L_1, L_2$ and an affine function $A$ of $\gf_{p^n}$ such that
$F_2=L_1\circ F_1\circ L_2+A$. Furthermore, $F_1$ and $F_2$ are called \textit{Carlet-Charpin-Zinoviev}-equivalent, in brief, (CCZ)-equivalent,
if there exists an affine permutation $\phi$ of $\gf_{p^n}\times\gf_{p^n}$ such that
$\phi\left(\{(x, F_1(x)) : x\in\gf_{p^n}\}\right)= \{(x, F_2(x)) : x\in\gf_{p^n}\} $.
It is well known that CCZ-equivalence implies EA-equivalence, but not vice versa.
Interested readers may refer to \cite{Cbook} for more details.

\subsection{Group rings and character theory}

Group rings and character theory of finite fields are useful tools to study
functions defined on $\gf_{p^n}$ and their related combinatorial objects.
We briefly review some definitions and results. For more details on group rings and character theory,
please refer to \cite{passman} and \cite{lidl} respectively.
In the following, we assume ${\cal G}$ is a finite Abelian group.
The group algebra $\mathbb{C}[{\cal G}]$ consists of all formal sums
$\sum\limits_{g\in {\cal G}}a_gg, a_g\in\mathbb{C}$.
We define component-wise addition
\begin{eqnarray*}
  \sum\limits_{g\in {\cal G}}a_gg + \sum\limits_{g\in {\cal G}}b_gg = \sum\limits_{g\in {\cal G}}(a_g+b_g)g,
\end{eqnarray*}
and multiplication by
\begin{eqnarray*}
	\sum\limits_{g\in {\cal G}}a_gg\cdot \sum\limits_{g\in {\cal G}}b_gg = \sum\limits_{g\in {\cal G}}\left(\sum\limits_{h\in {\cal G}}a_h\cdot b_{gh^{-1}}\right)g.
\end{eqnarray*}

A subset $S$ of ${\cal G}$ is identified with the group ring element $\sum_{s\in S}s$ in
$\mathbb{C}[{\cal G}]$, which is also denoted by $S$ (by abuse of notation).
For $A=\sum_{g\in {\cal G}}a_gg$ in $\mathbb{C}[{\cal G}]$ and
$t$  an integer, we define
$A^{(t)}:=\sum_{g\in {\cal G}}a_gg^t$.

A character $\chi$ of a finite Abelian group ${\cal G}$ is a homomorphism
from ${\cal G}$ to $\mathbb{C}^*$. A character $\chi$ is called
\textit{principal} if $\chi(c)=1$ for all $c\in {\cal G}$, otherwise it is
called \textit{non-principal}. All characters form a group which is denoted
by $\widehat{{\cal G}}$. This {\em character group} $\widehat{{\cal G}}$ is isomorphic to ${\cal G}$.
We denote its unity by $\chi_0$ and the unity of ${\cal G}$ by $1_{\cal G}$. The action of any character $\chi$ is extended to $\ec[{\cal G}]$ by $\chi(\sum_{g\in {\cal G}} a_gg)=\sum_{g\in {\cal G}}a_g\chi(g)$.
The following well-known Inversion formula is very useful to us.
\begin{lemma}[Inversion Formula]
  \label{invfor}
  Let $D=\sum\limits_{g\in {\cal G}} a_gg\in\ec[{\cal G}]$. Then
  $
     a_g =\frac{1}{|{\cal G}|} \sum\limits_{\chi\in \widehat{{\cal G}}}\chi(D)\chi(g^{-1}).
  $
\end{lemma}

The following lemma is an application of the Inversion formula.
\begin{lemma}\label{result2}
  \label{groupringeq}
  Let $D_1, D_2\in\mathbb{C}[{\cal G}]$ be two group ring elements. Then  $D_1=D_2$
  if and only if $\chi(D_1)=\chi(D_2)$ for all characters of ${\cal G}$.
\end{lemma}

Finally, for the finite field $\mathbb{F}_{p^n}$, define $\chi_1:\ef_{p^n}\rightarrow\mathbb{C}$ as
$\chi_1(x):=\zeta_p^{\tr(x)}$ for all $x\in \mathbb{F}_{p^n}$. Then
$\chi_1$ is an additive character of $\mathbb{F}_{p^n}$. Moreover,
every additive character $\chi$ is of the form $\chi_b$ $(b\in
\mathbb{F}_{p^n})$, where $\chi_b$ is defined by
$\chi_b(x)=\chi_1(bx)$ for all $x\in \mathbb{F}_{p^n}$.

\subsection{Partial difference sets and strongly regular graphs}
\label{2.3}

Let ${\cal G}$ be a multiplicative group of order $v$. A $k$-subset $D$ of ${\cal G}$ is called a
$(v,k,\lambda,\mu)$ \textit{partial
difference set} (PDS) if each non-identity element in $D$ can be represented as $gh^{-1}\ (g,h\in D, g \ne h)$
in exactly $\lambda$ ways, and each non-identity element in ${\cal G}\backslash D$ can be represented as
$gh^{-1}\ (g,h\in D, g \ne h)$ in exactly $\mu$ ways. We shall always assume that the identity element
$1_{\cal G}$ of ${\cal G}$ is not contained in $D$. Particularly, $D$ is called \textit{regular} if, denoting $D^{(-1)}:=\{d^{-1};\, d\in D\}$, we have $D^{(-1)}=D$.
Using the group ring language, a $k$-subset
$D$ of ${\cal G}$ with $1_{\cal G}\not\in D$ is a $(v,k,\lambda,\mu)$-PDS if and only if the following
equation holds:
\begin{equation}
	\label{PDSequation}
	DD^{(-1)}=(k-\mu)1_{\cal G}+(\lambda-\mu)D+\mu {\cal G}.
\end{equation}
Particularly, for a PDS, when $\lambda=\mu$, this reduces to the so-called
difference set. A $k$-subset $D$ of ${\cal G}$ is called a $(v,k,\lambda)$ \textit{difference set} (DS)
if each nonidentity element of ${\cal G}$ can be represented in the form $d_1d_2^{-1}\ (d_1,d_2\in D,
d_1\neq d_2)$ in exactly $\lambda$ ways. Similarly, using group ring notation, the subset $D$
is a $(v,k,\lambda)$ difference set if and only if
\begin{equation*}
DD^{(-1)}=k 1_{\cal G}+\lambda({\cal G}-1_{\cal G}).
\end{equation*}
By Lemma \ref{groupringeq}, we have the following result to show a $k$-subset $D$ is a PDS.
This result can be found in \cite{ma-pds}.
\begin{lemma}[\cite{ma-pds}]
  \label{pdsiff}
  Let $D$ be a group ring element of $\mathbb{C}[{\cal G}]$ with $|D|=k$.
  Then
  \begin{enumerate}[(i)]
  \item $D$ is a $(v,k,\lambda)$ difference set if and only if $\chi(DD^{(-1)})=k-\lambda$ for all
  non-principal character $\chi$ and $k^2=(k-\lambda) + \lambda v$.
  \item $D$ is a $(v,k,\lambda,\mu)$ partial difference set if and only if,
  for any nonprincipal character $\chi$ of ${\cal G}$,
  \begin{equation}
  \chi(DD^{(-1)})=k-\mu+(\lambda-\mu)\chi(D)
  \end{equation}and $k^2=(k-\mu)+k(\lambda-\mu)+\mu v$.\\
  If $D$ is regular then $\chi(D)^2=\chi(DD^{(-1)})$ and the former condition is equivalent to:
  \begin{equation*}m
	  \chi(D)= \frac{(\lambda-\mu)\pm \sqrt{(\mu-\lambda)^2-4(\mu-k)}}{2}.
  \end{equation*}

  \end{enumerate}
\end{lemma}

Combinatorial objects associated with partial difference sets are
strongly regular graphs. A graph $\Gamma$ with $v$ vertices is
called a $(v,k,\lambda,\mu)$ \textit{strongly regular graph} (SRG)
if each vertex is adjacent to exactly $k$ other vertices, any two
adjacent vertices have exactly $\lambda$ common neighbours, and any two non-adjacent vertices have exactly $\mu$
common neighbours.

Given a group ${\cal G}$ of order $v$ and a $k$-subset $D$ of ${\cal G}$ with $1_{\cal G}\not\in D$
and $D^{(-1)}=D$, the
graph $\Gamma=(V, E)$ defined as follows is called
the \textit{Cayley graph} generated by $D$ in ${\cal G}$:
\begin{enumerate}[(1)]
	\item The vertex set V is ${\cal G}$;
	\item Two vertices $g,h$ are joined by an edge if and only if $gh^{-1}\in D$.
\end{enumerate}
The following lemma points out the relationship between SRGs and PDSs.
\begin{lemma}[\cite{ma-pds}]
	\label{pdsandsrg} Let $\Gamma$ be the Cayley graph generated by a
	$k$-subset $D$ of a multiplicative group ${\cal G}$ with order
	$v$. Then $\Gamma$ is a $(v,k,\lambda,\mu)$ strongly regular graph
	if and only if $D$ is a $(v,k,\lambda,\mu)$-PDS with $1_{\cal G}\not\in
	D$ and $D^{(-1)}=D$.
\end{lemma}

Strongly regular graphs (or partial difference sets) with
parameters $(n^2,r(n+\varepsilon),-\varepsilon n+r^2+3\varepsilon
r,r^2+\varepsilon r)$ are called of \textit{Latin Square type}  if
$\varepsilon=-1$, and  of \textit{negative Latin Square type} if
$\varepsilon=1$. There are many constructions of SRGs of Latin
square type (see Lemma \ref{PCP} in Section 6), but only a few constructions of negative Latin square
type are known. We will show that such graphs may be obtained from quadratic 
zero-difference $\delta$-balanced functions.

\section{A new approach for constructing differentially uniform quadratic functions}

In this section, we first discuss the properties of differentially $\delta$-vanishing (defined below) functions,
and then use them to construct differentially $\delta$-uniform functions.

\subsection{A new notion related to differential uniformity, and its properties}
It is well-known that if a function $F:\ef_{p^n}\rightarrow\ef_{p^n}$ is quadratic, that is,
if all of its derivatives $F(x+a)-F(x)$ are affine,
then it is differentially $\delta$-uniform if and only if,
for every $a\neq 0$ in $\ef_{p^n}$, the related homogeneous linear equation $F(x+a)-F(x)=F(a)-F(0)$
has at most $\delta$ solutions. We have then:

\begin{proposition}
	\label{p1}
	Let $F:\ef_{p^n}\rightarrow\ef_{p^n}$ be a quadratic function.
	Then $F$ is differentially $\delta$-uniform if and only if,
	for every $a\neq 0$, there exists $b\in \ef_{p^n}$ such that the
	equation $F(x+a)-F(x)=b$ has at least one and at most $\delta$ solutions in $\ef_{p^n}$.
\end{proposition}
Indeed, the condition is clearly necessary (take $b=F(a)-F(0)$),
and it is also sufficient since the two linear equations
$F(x+a)-F(x)=F(a)-F(0)$ and $F(x+a)-F(x)=b$ have then the same number of solutions, since we know that they have both solutions and that they have the same linear part.

An interesting case of application is when $b$ can be taken equal to 0 for every $a\ne 0$. 

\begin{definition}
Let $\delta$ be some positive integer. A function $F$ is called
	{\em differentially $\delta$-vanishing}
	if for every $a\neq 0$, the equation $F(x+a)-F(x)=0$ has at least one and at most $\delta$ solutions in $\ef_{p^n}$; it is called  {\em zero-difference $\delta$-balanced} if the equation has $\delta$ solutions for every $a\neq 0$  in $\ef_{p^n}$.
\end{definition}

The notion of differentially $\delta$-vanishing function is new; the closely related notion of zero-difference
$\delta$-balanced function has been introduced in \cite{Ding-Tan}.

\begin{remark}\label{r1} Of course, zero-difference $\delta$-balanced implies differentially $\delta$-vanishing.
\end{remark}

For $\delta=1$ (in odd characteristic) and $\delta=2$ in characteristic 2,
the two notions ``zero-difference $\delta$-balanced'' and ``differentially
$\delta$-vanishing'' coincide.\\

According to Proposition \ref{p1}, we have:
\begin{proposition}\label{p2}
	Any quadratic differentially $\delta$-vanishing function is differentially $\delta$-uniform.
\end{proposition}

We shall see that many known quadratic differentially $\delta$-uniform (in particular, PN and APN)
functions, are in fact zero-difference $\delta$-balanced or differentially $\delta$-vanishing (with $\delta=1$ and 2).
The notion of zero-difference $\delta$-balanced and differentially vanishing functions gives a simpler reason
why these functions are differentially $\delta$-uniform. It will also lead to new constructions.

We give now a construction of functions $F$ which are differentially vanishing, and more precisely zero-difference balanced.

\begin{proposition}
	Let $d$ be any positive integer and $e=gcd(d,p^n-1)$.
	The function $x^d$ on $\ef_{p^n}$ is zero-difference $(e-1)$-balanced.
\end{proposition}
\begin{proof}
For every $x\in \ef_{p^n}$, we have  $(x+a)^d=x^d$ if and only if $(x+a)^e=x^e$,
which is equivalent to $x+a=wx$ for some $e$-th root of unity $w$ in $\ef_{p^n}$.
There are $e$ such $w$. One of them is $w=1$; the equation $x+a=wx$ is then impossible.
Any other $e$-th root of unity $w$ gives one distinct solution $x=\frac {a}{w-1}$.
\end{proof}

Of course, composing a differentially $\delta$-vanishing function $F$ on the left by a
function which is injective on the image set $\{F(x),x\in \ef_{p^n}\}$ of $F$ gives a
differentially $\delta$-vanishing function.
\begin{corollary}\label{dto1map}
        Let $d$ be any positive integer and $G$ a function from $\gf_{p^n}$ to $\gf_{p^n}$ such that $G(x^d)$ is quadratic.
	Let $e=gcd(d,p^n-1)$ and $C_d=\{x^d : x\in \gf_{p^n}\}=C_e$.
        Then function $F(x)=G(x^d)$ is a differentially $(e-1)$-uniform function if and only if
	the restriction $G\mid _{C_d}$ of $G$ to $C_d$ is an injection.
\end{corollary}


\begin{example}
As an application, we revisit the APNness of the well-known APN function
$F(x)=x^3+\tr(x^9)$ over $\gf_{2^n}$, in the case where $n$ is even. We have $F(x)=G(x^3)$, where $G(x)=x+\tr(x^3)$.
Function $G$ is a permutation polynomial. Indeed, if $G(x)=G(x+a) $
for some nonzero $a$, we get $a+\tr(a^3+a^2x+ax^2)=0$.
Since $a\neq 0$ and $a\in \ef_2$, then $a=1$, and therefore $a+\tr(a^3+a^2x+ax^2)=1$, a contradiction.
Therefore, $G$ is a permutation and then $F$ is an APN function.
\end{example}

Several remarks on Corollary \ref{dto1map} can be made:
\begin{remark}
	\begin{enumerate}[(1)]
		\item When $d=2$ and $p$ is odd, if the function $F(x)=G(x^2)$ is quadratic and if
		      $G\mid _{C_2}$ is an injection,  then $F$ is a perfect nonlinear function.
	              It is further proven in \cite{gohar-pott, weng} that all quadratic
		      PNs are EA-equivalent to functions of the form $G(x^2)$ where $G\mid_{C_2}$ is an injection.
                \item When $d=3$, the function $F(x)=G(x^3)$, if it is quadratic and if $G\mid _{C_3}$
 	              is an injection, is an APN function. The condition ``quadratic`` seems necessary; there are many examples of permutations $G$ such that $G(x^3)$, non-quadratic, is not APN; for instance there are many examples of power functions $x^d$
	    over $\gf_{2^n}$, $n$ even,         where $d$ is co-prime with $2^n-1$ and such that $x^{3d}$ is not APN.
	              \item However, there are also examples of non-quadratic APN functions which are zero-difference $2$-balanced. 	              This is the case of all APN power functions over $\gf_{2^n}$, $n$ even,  since Dobbertin has shown that for these functions $x^d$, we have $gcd(d,2^n-1)=3$ (see a recall of his proof in \cite{Cbook}).  The Kasami functions $x^{2^{2k}-2^k+1}$, where $k$ is co-prime with $n$ and  the Dobbertin function $x^{2^{\frac{4n}{5}}+2^{\frac{3n}{5}}+2^{\frac{2n}{5}}+2^{\frac{n}{5}}-1}$ where $n$ is divisible by $10$ are such functions. This fact has also been observed in the bottom of \cite[page 8]{carlet-dcc2010}.
	              
	              \item There
	              are many quadratic APN functions which are not of the form $G(x^3)$
	              (and which are not differentially 2-vanishing); for instance,
	              by checking the $2,275$ quadratic APN functions listed in \cite{yu-APN} and \cite{tan-APN},
                      we found that only $18$ of them are of the form $G(x^3)$. Note that every function of
	              the form $G(x)=\sum_{i} a_ix^{\frac{2^{k_i}+2^{l_i}}{3}}$ where, for every $i$,
	              $k_i$ is an odd number and $l_i$ is an even number (and then $\frac{2^{k_i}+2^{l_i}}{3}$ is an integer),
	              is such that $G(x^3)$ is quadratic.
	
	\end{enumerate}
\end{remark}

\subsection{Further properties}
In Corollary \ref{dto1map}, function $G$ does not need to be bijective; it just needs to be injective 
on $C_d$. But given such $G$, we can always find a permutation $G'$ on $\gf_{p^n}$
such that $G'\mid _{C_d}=G\mid _{C_d}$ and therefore $F(x)=G'(x^d)$.
Indeed, any function coinciding with $G$ on $C_d$ and mapping injectively the
complement of $C_d$ onto the complement of $G(C_d)$ can be taken for $G'$. In more precise setting:

\begin{proposition}
	Let $d$ be a positive integer, $e=gcd(d,p^n-1)$ and $G$ a function defined on
	$\gf_{p^n}$ such that $G\mid_{C_d}$ is an injection. Let $h:\gf_{p^n}\rightarrow\gf_p$
	be the characteristic function of $C_d$:
	\begin{equation*}
		h(x)= 1-\left( x^{\frac{2(p^n-1)}{e}} - x^{\frac{p^n-1}{e}}\right)^{p^n-1},
	\end{equation*}
        (satisfying $h(x)=1$ if $x\in C_d=C_e$, and $h(x)=0$ otherwise).
        There exists a function $T(x)$ on $\gf_{p^n}$
	such that the function $G'$ defined by
	\begin{equation}
		\label{relatedpp}
		G'(x)=h(x)G(x)+(1-h(x))T(x)
        \end{equation}
	is a permutation and satisfies $G'\mid _{C_d}=G\mid _{C_d}$, that is, $G(x^d)=G'(x^d)$, for all $x$.
\end{proposition}

\begin{proof}
First, we show that $h$ is indeed the characteristic function of $C_d$: if $x\in C_d$, that is,
$x\in C_e$, then $x^{\frac{p^n-1}{e}}\in \{0,1\}$ implies
$x^{\frac{2(p^n-1)}{e}} - x^{\frac{p^n-1}{e}}=0$ and then $h(x)=1$;
and if $x\not\in C_d$, that is, $x\not\in C_e$, then $x^{\frac{p^n-1}{e}}\not\in \{0,1\}$
implies $x^{\frac{2(p^n-1)}{e}} - x^{\frac{p^n-1}{e}}\neq 0$ and then $h(x)=0$.
Next, clearly there exists a function $T:\gf_{p^n}\rightarrow\gf_{p^n}$ whose restriction
$T|_{\gf_{p^n}\setminus C_d}$ is a bijection from
$\gf_{p^n}\setminus C_d$ to $\gf_{p^n}\setminus \mbox{Im}(G|_{C_d})$, then $G'$ defined in (\ref{relatedpp}) is
a permutation. Indeed, this is clear from $G'(x)=G(x)$ for $x\in C_d$ and $G'(x)=T(x)$ for $x\not\in C_d$,
and the properties of $G$ and $T$. This completes the proof.
\end{proof}

A natural question one may ask is the following:

\begin{problem}
Do all differentially $\delta$-vanishing quadratic functions have the form $G(x^d)$ with
$\delta=gcd(d, p^n-1)-1$ and where $G\mid_{C_d}$ is an injection?
\end{problem}
We leave this problem open.\\

In the following we characterize the zero-difference balanced functions
by their Walsh transform.
\begin{proposition}
	Let $F$ be a function on $\gf_{p^n}$. Then $F$ is zero-difference $\delta$-balanced if and only if
        \begin{equation}
		\label{walshcond1}
		\sum\limits_{b\in\gf_{p^n}}\mathcal{W}_F(a,b)\overline{\mathcal{W}_F(a,b)}=\left\{
			  \begin{array}{ll}
				  p^{2n} - \delta p^n,                        &\mbox{if}\ a\neq 0, \\
				  (\delta+1)p^{2n}-\delta p^n,                &\mbox{if}\ a=0,
	                   \end{array}
                          \right.
	\end{equation}
        where $\overline{w}$ denotes the complex conjugate of the complex number $w$.
        If $F$ is 
        differentially $\delta$-uniform, then the condition when $a=0$:	
        \begin{equation}
		\label{walshcond2}
	        \sum\limits_{b\in\gf_{p^n}}\mathcal{W}_F(0,b)\overline{\mathcal{W}_F(0,b)}=(\delta+1)p^{2n}-\delta p^n	
        \end{equation}	
        is sufficient (and so is necessary and sufficient).
\end{proposition}

\begin{proof}
        Function $F$ is zero-difference balanced if and only if the  numerical function
	$$ \sigma: u\mapsto |\{x : x \in \gf_{p^n} | F(x+u) - F(x)=0 \}| $$
	takes constant value $\delta$ at every $u\neq 0$ and takes value $p^n$ at 0.
	This is equivalent to the fact that the numerical function
	$$ \sigma': u\mapsto \sum\limits_{b,x\in\gf_{p^n}}\zeta_p^{\tr(b(F(x+u) - F(x)))} $$
	takes constant value  $\delta p^{n}$ at every $u\neq 0$ and takes value $p^{2n}$ at 0.
        By the properties of the Fourier transform (see e.g. \cite{Cbook1}), this is equivalent
	to the fact that the Fourier transform of $\sigma'$, that is,
	\begin{eqnarray*}
	   \widehat{\sigma'}(a) &=&  \sum_{u\in \gf_{p^n}}\sigma'(u)\zeta_p^{\tr(au)}=
	                             \sum\limits_{b,x,u\in\gf_{p^n}}\zeta_p^{\tr(b(F(x+u) - F(x)))+Tr(au)} \\
	                        &=&  \sum\limits_{b,x,y\in\gf_{p^n}}\zeta_p^{\tr(b(F(y) - F(x)))+Tr(a(y-x))}=
	                             \sum\limits_{b\in\gf_{p^n}} \mathcal{W}_F(a,b)\overline{\mathcal{W}_F(a,b)}
	\end{eqnarray*}
	takes constant value  $p^{2n}-\delta p^n$ at every $a\neq 0$ and takes value
	$\delta p^n(p^n-1)+p^{2n}=(\delta+1)p^{2n}-\delta p^n$ at 0. Therefore, $F$ is differentially
	$\delta$-uniform if and only if (\ref{walshcond1}) holds.
\\Furthermore, if $F$ is 
differentially $\delta$-uniform, then we know that the number
	of solutions of $F(x+u) - F(x)=0$ is not larger than $\delta$. Hence, we need only to show that it is never strictly smaller than $\delta$ and this is characterized by $\widehat{\sigma'}(0)=\delta p^n(p^n-1)+p^{2n}$, that is by (\ref{walshcond2}).
\end{proof}

\section{New families of quadratic zero-difference balanced functions over $\gf_{p^n}$}
In this section, we present new families of quadratic zero-difference balanced functions.
In Section 6 below, we will demonstrate that some of these functions give rise to new negative Latin square
type strongly regular graphs.
As mentioned in the Introduction, the classes of ZDB functions constructed in 
\cite{ZDB-Ding,QiZhou-ZDB,Han-ZDB,Ding-ZDB,ZTWY} are on a cyclic group except
the one in \cite[Theorem 1]{Ding-ZDB}. However, one can easily see our constructions
below are different from it by comparing the parameters of the ZDB functions.
Therefore, the ZDB functions presented in this section are new.

\subsection{A class of quadratic zero-difference 2-balanced functions over $\gf_{2^n}$}\label{4.1}

Recall that quadratic zero-difference 2-balanced functions over $\gf_{2^n}$ are APN functions. In this section we use Corollary \ref{dto1map} to characterize a family of APN functions
of the form $x^3 + \sum\limits_{i=1}^\ell\alpha_i\tr(\beta_i x^3 + \gamma_i x^9) $ defined on
$\gf_{2^n}$ with $n$ even. We begin with the subclass of those functions of the form $x^3+\alpha\tr(\beta x^3 + \gamma x^9)$.
When $n=8$, by choosing proper $\alpha,\beta,\gamma$, it generalizes a sporadic
example discovered in \cite{edel-pott}.

\begin{proposition}
	\label{newapn}
	Let $G$ be a function on $\gf_{2^n}$  defined by
	$$ G(x) = x + \alpha\tr(\beta x + \gamma x^3), $$
        where $\alpha, \beta, \gamma$ are elements in $\gf_{2^n}$, $\alpha\neq 0$ and $n$ is an even integer.
        Then $G$ is a permutation polynomial of $\gf_{2^n}$ if and only if
        (i) $\gamma=0$ and $\tr(\beta\alpha)=0$, or (ii) $\gamma \alpha^3=1$ and $\tr(\beta\alpha)=0$.
       If one of these two conditions is satisfied, the function $F(x)=G(x^3)= x^3 + \alpha\tr(\beta x^3 + \gamma x^9)$
        is a quadratic APN function.
\end{proposition}

\begin{proof}
        Function $G$ is a permutation polynomial (PP) if and only if, for every $a\in \gf_{2^n}^*$, the equation $G(x+a)+G(x)=0$ has
	no solution. This equation is equivalent to:
	\begin{equation}
	   \label{thm2eq1}
	    a\alpha^{-1} + \tr\left(\beta a + \gamma a^3+(\gamma a^2+(\gamma a)^{2^{n-1}})x\right)=0.
	\end{equation}
        Since $a\neq 0$,  the above equation holds only if $a=\alpha$ and
	\begin{equation}
	   \label{thm2eq2}
           1 + \tr\left(\beta\alpha + \gamma \alpha^3+(\gamma \alpha^2 + (\gamma \alpha)^{2^{n-1}})x  \right)=0.
        \end{equation}
	If $\gamma \alpha^2 + (\gamma \alpha)^{2^{n-1}}\neq 0$, then it is clear that (\ref{thm2eq2}) always
	has solutions. If $\gamma \alpha^2 + (\gamma \alpha)^{2^{n-1}}=0$, that is, if $\gamma=0$ or
	$\gamma \alpha^3=1$, then (\ref{thm2eq2}) has solutions if and only if
	$\tr(\beta\alpha + \gamma \alpha^3)=1$, that is, $\tr(\beta\alpha)=1$. This completes the proof.
\end{proof}

\begin{remark}
	Some remarks on the CCZ-inequivalent APN functions generated by Proposition \ref{newapn}:
	\begin{enumerate}[(1)]
	   \item For $\gamma=0$ and $\tr(\beta\alpha)=0$,  $G(x)$ is a linear permutation, and hence the APN function $F(x)=G(x^3)$ is CCZ-equivalent
	         to the Gold APN function $x^3$.
	   \item For $\beta=0$ and $\gamma\alpha^3=1$, we have $F(x)=\alpha\left(\frac{x^3}{\alpha}+\tr\left(\left(\frac{x^3}{\alpha}\right)^3\right)\right)$. On $\gf_{2^8}$ and $\gf_{2^{10}}$, we checked that taking for $\alpha$ a primitive element of $\gf_{2^2}$ gives a function CCZ-inequivalent to $x^3+\tr(x^9)$ and $x^3$.
	   \item By exhaustive search of all $\alpha,\beta,\gamma$ on $\gf_{2^8}$, there are only three (up to equivalence) aforementioned
                 APN functions found. Furthermore, on $\gf_{2^{10}}$, by partly search $\alpha,\beta,\gamma$,
	         we get three APN functions, $x^3, x^3+\tr(x^9)$ and one which is not in any known infinite families.
	\end{enumerate}
\end{remark}
\subsection{Obtaining APN functions from known ones}
We give now a method to generate APN functions obtained by Corollary \ref{dto1map}  from known ones.

\begin{lemma}
	\label{injlemma}
	Let $G:\gf_{2^n}\rightarrow\gf_{2^n}$ be a function satisfying that $G|_{C_3}$ is an injection,
        and $h:\gf_{2^n}\rightarrow\gf_2$ be a Boolean function.
        Let $\gamma\in\gf_{2^n}$ be a nonzero constant.
        Then the function $H:\gf_{2^n}\rightarrow\gf_{2^n}$ defined by $H(x)=G(x)+\gamma h(x)$ has also injective restriction to $C_3$ if and only if
        \begin{equation}
            \label{injcondition}
            h(x^3)+h(y^3)=0\ \mbox{ holds for any}\ x,y\ \mbox{satisfying}\ G(x^3)+G(y^3)=\gamma .
        \end{equation}
       The set
        $ S_{G,\gamma}:=\{h\in\mathcal{BF}_n| G(x)+\gamma h(x)\ \mbox{is an injection on}\ C_3\} $  is a subspace of $(\mathcal{BF}_n, +)$,
        where $\mathcal{BF}_n$ is the set of all Boolean functions.
\end{lemma}

\begin{proof}
First, we assume that condition (\ref{injcondition}) holds and we show that $H\mid_{C_3}$ is an injection, that is, $H(x^3)+H(y^3)=0$ if and only if $x^3=y^3$.
Expanding $H$ we get $G(x^3)+G(y^3)=\gamma(h(x^3)+h(y^3))$. If $h(x^3)+h(y^3)=0$,
we have $G(x^3)+G(y^3)=0$ and hence $x^3=y^3$ by the injectivity of $G$. Otherwise,
$G(x^3)+G(y^3)=\gamma$ when $h(x^3)+h(y^3)=1$, but this is not possible by condition (\ref{injcondition}).
Conversely, assume that $H\mid_{C_3}$ is an injection and there exist $x_0,y_0$ such that
$G(x_0^3)+G(y_0^3)=\gamma$ and $h(x_0^3)+h(y_0^3)=1$, this would lead to $H(x_0^3)+H(y_0^3)=0$ and $x_0^3\ne y_0^3$,
which is not possible as $H\mid_{C_3}$ is an injection.
\\
Finally, it is easy to see that, for $h_1,h_2\in S_{G,\gamma}$, $h_1+h_2$ satisfies
condition (\ref{injcondition}) and hence $h_1+h_2\in S_{G,\gamma}$, which implies that
$S_{G,\gamma}$ is a subspace.
\end{proof}

Lemma \ref{injlemma} above provides an algorithmic method to find functions $H=G+\gamma h$
with the property that $H\mid_{C_3}$ is an injection from known such functions $G$. This can be done by solving a
linear system of  equations in $h$: we regard $h$ as
the vector of length $2^n$ in the last column of its truth table (by abuse of notation, we still denote this vector  by $h$); a Boolean function $h$ belongs then to the vector space $S_{G,\gamma}$ if and only if
$ R\times h^T=0$, where $\times$ denotes the matrix product and $R$ is the matrix whose term at row indexed by an ordered pair $\{x^3,y^3\}$ such that $G(x^3)+G(y^3)=\gamma$ and at column indexed by $u\in \gf_{2^n}$ equals 1 if $u=x^3$ or $u=y^3$ and equals 0 otherwise.

If $G(x^3)$ is quadratic, then $H(x^3)$ will be quadratic if and only if $h(x^3)$ is quadratic. As already observed, this is achieved when the degree 2 part of $h(x)$ has the form $\sum_{i} a_ix^{\frac{2^{k_i}+2^{l_i}}{3}}$ where, for every $i$,
	              $k_i$ is an odd positive number and $l_i$ is an even positive number, and where $1<\frac{2^{k_i}+2^{l_i}}{3}<2^n-1$ (since $\frac{2^{k_i}+2^{l_i}}{3}=2^n-1$ is impossible, unless $n\le 2$). Note that such $h$ is Boolean when $a_{j}=a_i^2$ for every $i,j$ such that $l_j=k_i+1$ and $k_j=l_i+1$, and for every $i,j$ such that $k_i=l_j+n-1$ and $l_i=k_j+n-1$.


\subsection{Zero-difference $p$-balanced functions over $\gf_{p^n}$}

We present quadratic zero-difference $p$-balanced functions over
$\gf_{p^n}$ in this section. Note that such functions are differentially
$p$-uniform functions by Remark \ref{r1} and Proposition \ref{p2}.
The following well-known lemma will be useful for the construction. We give a proof for self-completeness.

\begin{lemma}
	\label{xxp}
	Let $t\in\gf_{p^n}$ with $n$ even, the equation $x+x^p=t$ has solutions on $\gf_{p^n}$ if and only if
	$ \sum\limits_{i=0\atop i\ \tiny\mbox{even}}^{n-2} t^{p^i}\in\gf_p$. The number of solutions is then $p$.
\end{lemma}
\begin{proof}
	If the equation $x+x^p=t$ has a solution $x\in \gf_{p^n}$,  then we have
	$\tr(x)=\sum\limits_{i=0\atop i\ \tiny\mbox{even}}^{n-2} t^{p^i}\in\gf_p$. All solutions of the equation are then $x+i$, $i\in \gf_p$.\\
	Conversely, assume that $ \sum\limits_{i=0\atop i\ \tiny\mbox{even}}^{n-2} t^{p^i}\in\gf_p$.  Let $\beta$ be a solution of the equation $x+x^p=t$ in some extension
	field of $\gf_{p^n}$. Then:
	\begin{eqnarray*}
             \beta^{p^n}-\beta &=& (\beta^{p^n}+\beta^{p^{n-1}}) - (\beta^{p^{n-1}}+\beta^{p^{n-2}}) + \ldots + (\beta^{p^2}+\beta^p) - (\beta^p+\beta) \\
                               &=& (\beta^p+\beta)^{p^{n-1}} - (\beta^p+\beta)^{p^{n-2}} + \ldots + (\beta^p+\beta)^{p} - (\beta^p+\beta) \\
                               &=& -\sum\limits_{i=0\atop i\ \tiny\mbox{even}}^{n-2} t^{p^i} +
			       \left(\sum\limits_{i=0\atop i\ \tiny\mbox{even}}^{n-2} t^{p^i}\right)^p =0.\\
		       \end{eqnarray*}
	Therefore we get $\beta\in\gf_{p^n}$ and this completes the proof.
\end{proof}

Now we are ready to present the main result of this section, which is an extension of the results of Section \ref{4.1} to general characteristic.

\begin{theorem}
	\label{newfun}
	Let $F(x)=x^{p+1}+\alpha\tr(\beta x^{p+1} + \gamma x^{p^3+1})$ defined on $\gf_{p^n}$,
	where $\alpha, \beta, \gamma\in\gf_{p^n}$.
	Then
        \begin{enumerate} [(i)]
		\item when $n=4$, function $F$ is a zero-difference $p$-balanced
	              function if and only if $\tr(\alpha\beta+\gamma\alpha^{p^3})\neq -1$.
                \item when $n=6$,  if $\gamma^{p^3-1}=-1$ and
		      $-1-\tr(\alpha\beta)\neq 0$, then function $F$ is a zero-difference $p$-balanced
		      function.
	\end{enumerate}
\end{theorem}
\begin{proof}
	To prove that $F$ is a zero-difference $p$-balanced, we need to show that, for any nonzero
	$a\in\gf_{p^n}$,
	the equation $\Delta_F(ax)=F(ax+a)-F(ax)=0$ has exactly $p$ solutions. \\
	Expanding $\Delta_F(ax)$ we have
	\begin{equation}
		\label{deltaeq}
		-a^{p+1}(1+x+x^p)=\alpha\tr\left(\beta a^{p+1}(1+x + x^p)+\gamma a^{p^3+1}(1+x+x^{p^3})\right).
	\end{equation}
	Clearly the above equation holds if and only if, for some $k\in\gf_p$, we have: $$\left\{\begin{array}{l}1+x+x^p=ka^{-p-1}\alpha\\-k=\tr\left(\beta a^{p+1}(1+x + x^p)+\gamma a^{p^3+1}(1+x+x^{p^3})\right)\end{array}\right..$$
	In the case that $k=0$, corresponding to $x+x^p=-1$,  we have
	$1+x+x^{p^3}=\left((1+x+x^p)-(1+x^p+x^{p^2})+(1+x^{p^2}+x^{p^3})\right)=0$, and
	then the second equation in the system above gives $0=0$, which implies that all solutions of $x+x^p=-1$
	are the solutions of (\ref{deltaeq}). Note also that the equation $x+x^p=-1$ has $p$ solutions
	in $\gf_{p^4}$ (resp. $\gf_{p^6}$) by Lemma \ref{xxp}.
	Therefore, $F$ is zero-difference $p$-balanced if and only if, $\alpha=0$ or,
	for any $k\neq 0$, the system above has no solution. The second equation in the system is equivalent to the following equation, obtained by replacing $1+x+x^p$ by its value from the first equation, using that $\tr(ku)=k\tr(u)$ for every $u\in \gf_{p^4}$ (resp. $\gf_{p^6}$), and dividing by $k$:
	\begin{equation}
		\label{eq2}
		-1-\tr(\alpha\beta)=\tr\left(\gamma(a^{p^3-p}\alpha- a^{p^3+1-p^2-p}\alpha^p + a^{1-p^2}\alpha^{p^2} ) \right).
	\end{equation}
$(i)$ In the case $n=4$, by Lemma \ref{xxp}, the equation $x+x^p=ka^{-p-1}\alpha-1$ has solutions in
	$\gf_{p^4}$ if and only if
	$\ell=a^{-p-1}\alpha + a^{-p^3-p^2}\alpha^{p^2}$ belongs to $\gf_p$, and therefore also equals $a^{-p^2-p}\alpha^p+a^{-1-p^3}\alpha^{p^3}$.
        Now, considering the right hand side of (\ref{eq2}), we have
	\begin{eqnarray*}
		\mbox{RHS} &=& \tr\left( \gamma a^{p^3-p}\alpha+\gamma a^{1-p^2}\alpha^{p^2}  - \gamma \alpha^p  a^{p^3+1-p^2-p}\right)  \\
	 	&=&  \tr\left(\gamma a^{p^3-p}\alpha+\gamma a^{1-p^2}\alpha^{p^2} -\gamma a^{p^3+1}(\alpha^p a^{-p^2-p})  \right) \\
		&=&  \tr\left(\gamma a^{p^3-p}\alpha+\gamma a^{1-p^2}\alpha^{p^2} -\gamma a^{p^3+1}(\ell - \alpha^{p^3}a^{-1-p^3})  \right) \\
		&=&  \tr\left(\gamma a^{p^3-p}\alpha+\gamma a^{1-p^2}\alpha^{p^2} -\ell\gamma a^{p^3+1} +\gamma \alpha^{p^3}\right) \\
		&=&  \tr\left(\gamma a^{p^3-p}\alpha+\gamma a^{1-p^2}\alpha^{p^2} -(a^{-p-1}\alpha+a^{-p^3-p^2}\alpha^{p^2})\gamma a^{p^3+1}
		+\gamma\alpha^{p^3}\right) \\
		&=& \tr\left(\gamma a^{p^3-p}\alpha+\gamma a^{1-p^2}\alpha^{p^2}  - (\gamma\alpha a^{p^3-p}+\gamma\alpha^{p^2}a^{1-p^2})
		+\gamma \alpha^{p^3}  \right)=\tr(\gamma\alpha^{p^3}).
	\end{eqnarray*}
	Clearly, (\ref{eq2}) does not have then any solution in $\gf_{p^4}$ if and only if
	$\tr(\gamma\alpha^{p^3})\neq -1-\tr(\alpha\beta)$. \\
	
	$(ii)$ In the case $n=6$, by Lemma \ref{xxp}, the equation $x+x^p=ka^{-p-1}\alpha-1$ has solutions in $\gf_{p^6}$ if and only if
	$\ell=a^{-p-1}\alpha + a^{-p^3-p^2}\alpha^{p^2} + a^{-p^5-p^4}\alpha^{p^4}=
	a^{-p^2-p}\alpha^p+a^{-p^4-p^3}\alpha^{p^3} + a^{-1-p^5}\alpha^{p^5}\in\gf_p$.
        The right hand side of (\ref{eq2}) equals then:
	\begin{eqnarray*}
		\mbox{RHS} &=& \tr\left(\gamma\alpha a^{p^3-p}+\gamma\alpha^{p^2}a^{1-p^2}-\gamma\alpha^pa^{p^3+1-p^2-p}  \right) \\
		&=& \tr\left(\gamma\alpha a^{p^3-p}+\gamma\alpha^{p^2}a^{1-p^2}
		-\gamma a^{p^3+1}(\ell-a^{-p^4-p^3}\alpha^{p^3}-a^{-1-p^5}\alpha^{p^5})  \right)  \\
		&=& \tr\left( \gamma\alpha a^{p^3-p}+\gamma\alpha^{p^2}a^{1-p^2}-\ell\gamma a^{p^3+1}
		+\gamma\alpha^{p^3}a^{1-p^4}+\gamma\alpha^{p^5}a^{p^3-p^5} \right) \\
		&=& \tr\left( (\gamma\alpha^{p^3}+(\gamma\alpha)^{p^3})a^{1-p^4}+
		    (\gamma\alpha^{p^2}+(\gamma\alpha^{p^5})^{p^3})a^{1-p^2}-\ell\gamma a^{p^3+1}\right) \\
		    &=& \tr\left( (\gamma+\gamma^{p^3})(\alpha^{p^3}a^{1-p^4}+\alpha^{p^2}a^{1-p^2})
		    -\ell\gamma a^{p^3+1}\right).
	\end{eqnarray*}
	Since $\gamma^{p^3-1}=-1$
	and hence $\gamma^{p^3}=-\gamma$, from the above equation we have
	$\mbox{RHS}=\tr(-\ell\gamma a^{p^3+1})=-\ell\tr(\gamma a^{p^3+1})=-\ell (\gamma a^{p^3+1}-\gamma a^{p^3+1}+\gamma a^{p^3+1}-\gamma a^{p^3+1}+\gamma a^{p^3+1}-\gamma a^{p^3+1})=0$.
Hence, if  $-1-\tr(\alpha\beta)\neq 0$, (\ref{eq2}) does not have solutions and we complete the proof.
\end{proof}

\begin{remark}
1. The functions of Theorem \ref{newfun} are of the form $G(x^{p+1})$. They satisfy Corollary \ref{dto1map}. Showing this would result in a similar but slightly more complex proof.\\
2. The condition on $\gamma$ in Theorem \ref{newfun} (ii) is equivalent to: ``$\gamma^2\in \gf_{p^3}^*$ and $\gamma\not\in \gf_{p^3}^*$''.\\
3.  For $p=3$ and $n=4$, by visiting
exhaustively all $\alpha,\beta$ and $\gamma$ in $\gf_{p^n}$ satisfying the condition in Theorem \ref{newfun} (i),
we get only one class under CCZ equivalence, namely $x^4$. It is interesting to mention that Theorem \ref{gengraph} in Section 5 will show that the image of any quadratic zero-difference $3$-balanced function on $\gf_{3^4}$
is a $(81,20,1,6)$-SRG. We know from \cite{81graph} that there is only one such graph up to isomorphism.\\
4.  For $p=5, n=4$, we could only perform partial search of coefficients $\alpha,\beta,\gamma\in\gf_{5^4}$;
we also get one class under CCZ inequivalence of quadratic zero-difference $5$-balanced function, namely $x^6$.\\
5. For $n=8,10$, by a computer search, we did not find other coefficients
$\alpha,\beta,\gamma$ such that $F$ in Theorem \ref{newfun} is a quadratic zero-difference $p$-balanced function, except for $\alpha=0$.\\
6. In Theorem \ref{newfun}(ii), there exist coefficients $\alpha,\beta,\gamma$ which do not satisfy the conditions in the Theorem, while the function $F$ is still a zero-difference $p$-balanced function. 
\end{remark}

\begin{remark}
It is worth mentioning that the ZDB functions considered in this paper, which are on the non-cyclic group  can be used to construct other objects. In \cite{ZDB-Ding}, Ding made use of zero-difference balanced functions to construct optimal constant composition codes (CCCs in short). An $(n, M, d, [w_0,\ldots,w_{q-1}])_q$ \textit{constant composition code} is a code over an Abelian group $\{b_0, b_1, \ldots, b_{q-1} \}$ with length $n$, size $M$, and minimum Hamming distance $d$ such that, in every codeword, the element $b_i$ appears exactly $w_i$ times for every $i$.

Let $A=\{a_0,\ldots,a_{n-1}\}$ and $B=\{b_0, \ldots, b_{\ell-1}\}$ be two Abelian groups, and let $\Pi$ be a function from $A$ to $B$. Define $w_i=|\{x\in A \;|\; \Pi(x)=b_i \}|$ for $0\leq i\leq \ell-1$.
Now define $\mathcal{C}_\Pi$ as 
\begin{equation}
\label{CCCdef}
\mathcal{C}_\Pi = \Big\{  \Big(\Pi(a_0+a_i), \ldots, \Pi(a_{n-1} + a_i)\Big) : 0\leq i\leq n-1   \Big\}.
\end{equation}
If $\Pi$ is a zero-difference $\delta$-balanced function, then $\mathcal{C}_\Pi$ is an $(n, n, n-\lambda, [w_0,\ldots,w_{\ell-1}])_\ell$ CCC over $B$.
In Proposition 6 and Theorem \ref{newfun}, we provide zero-difference $p$-balanced functions from $\gf_{p^n}$ to itself, where $n=2k$. By the above result, we may obtain optimal CCCs using these ZDB functions. In the following we determine the parameters of these CCCs. Note that we state the result in a more general form.

Let $F$ be a function from $\gf_{p^n}$  to itself with the form $F(x)=G(x^d)$, where $d\mid p^n-1$ and $G|_{C_d}$ is an injection. Assuming that $F(0)=0$. Denoting the preimage set of $C_d\setminus \{0\}$ by $G^{-1}(C_d^*)$.
Indexing the elements in $\gf_{p^n}$ as 
\[ \left[\; 0, x\in G^{-1}(C_d^*), x\not\in G^{-1}(C_d^*)\; \right].  \]
Then the code defined in \emph{(\ref{CCCdef})} is an optimal CCC with parameter
$$
\Big( p^n, p^n, p^n - (d-1), [ 1, \underbrace{ {d, \cdots, d} }_{(p^n-1)/d} , \underbrace{ {0, \cdots, 0}}_{(p^n- 1)(d-1)/d}  ] \Big)_{p^n}.
$$

The proof of the above claim is not difficult. Indeed, it is easy to see that $F$ is a zero-difference $(d-1)$-balanced function on $\gf_{p^n}$. Indeed, for any $a\neq 0$, the equation $F(x+a)-F(x)=0$ leads to $G((x+a)^d)=G(x^d)$, and therefore $(x+a)^d=x^d$ since $G|_{C_d}$ is an injection. Now we have that $x+a=wx$ for some $w\in\{ x\in\gf_{p^n} | x^d =1 \}$. Note that there are $d$ such $w$ and only for $w=1$ the equation $x+a=wx$ does not have solution (since $a\neq 0$), we have shown that $F$ is zero-difference $(d-1)$-balanced function, and therefore $\mathcal{C}_\Pi$ is a constant composition code. Now we determine the frequency set $\{ w_i : 0\leq i\leq p^n-1\}$, where $w_i = |\{ x\in\gf_{p^n} \;|\; F(x) = G(x^d) = b_i \}|$. It is clear to see that the equation $G(x^d)=b_i$ has solution if and only if $b_i\in G(C_d)$. First, it is clear that $w_0=1$ since from $G(x^d)=0$ we have that $G(x^d)=0=G(0)$, and then $x=0$ as $G|_{C_d}$ is an injection. Otherwise, if $b_i\in G(C_d)$ and $b_i\neq 0$ we have $w_{b_i}=d$. This is because from the equation $G(x^d)=b_i$ we get $x^d=G^{-1}(b_i)$ (using the assumption $G|_{C_d}$ is an injection), which has exactly $d$ solutions.   
We complete the proof.

\end{remark}

\section{Strongly regular graphs and quadratic zero-difference $p^t$-balanced functions}

In this section, we discuss the relationship between partial differential sets and
quadratic zero-difference $p^t$-balanced functions, which are of the form
$F(x)=G(x^{p^t+1})$ where $G|_{C_{p^t+1}}$ is an injection.
We recall first a well-known fact and we give a lemma which is used to prove Theorem \ref{gengraph} below.

\begin{lemma}[\cite{pottbook}]
	\label{rationalinteger}
	An algebraic integer $X\in\mathbb{Z}[\zeta_p]$ is a rational integer if and only if
	$\sigma_\alpha(X)=X$ for all $\alpha\in\gf_p$ with $\alpha\neq 0$, where
	$\sigma_\alpha\in \mbox{Gal}(\mathbb{Q}(\zeta_p)/\mathbb{Q})$ defined by
	$\sigma_\alpha(\zeta_p)=\zeta_p^\alpha$.
\end{lemma}

In the following, for any nonzero integer $a$, we use the notation
$\od(a)$ to denote the highest integer $t$ such that $2^t\mid a$ and $2^{t+1}\nmid a$.

\begin{lemma}
	\label{divide}
	Let $p$ be a prime and $t, n$ be two positive integers.
	  Let $a=\od(p^t+1), b=\od(p^t-1), c=\od(p^n-1)$ and $d=\od(p^n+1)$.
	Then:
	\begin{enumerate}[(i)]
	  \item $\gcd(p^t+1,p^n-1)=\delta_{t,n}\cdot\frac{p^{\gcd(2t,n)}-1}{p^{\gcd(t,n)}-1}$,
	    where $\delta_{t,n}=2^{\min(a,c)+\min(b,c)-\min(a+b,c)}\in \{1,2\}$.
	  Furthermore, $p^t+1\mid p^{n}-1$ if and only if $2t\mid n$;
	  \item
	      $\gcd(p^t+1,p^n+1)=\eta_{t,n}\cdot \frac{p^{\gcd(2t, 2n)}-1}{p^{\gcd(t, 2n)}-1}
	      \cdot\frac{p^{\gcd(t, n)}-1}{p^{\gcd(2t, n)}-1}$,
	    where $\eta_{t,n}=\frac{\eta_{t,n}'\delta_{t,2n}}{\delta_{t,n}}\in \{1,2\}$
	    and $\eta_{t,n}'= 2^{\min(a,c)+\min(a,d)-\min(a,c+d)}$.
	    Furthermore, $p^t+1\mid p^n+1$ if and only if
	    $n=\ell t$ for some odd integer $\ell$.
	
	  \item Let $n=2kt$ for some positive integer $k$. Then, for every positive integer $i$, we have
			$p^t+1\mid p^{n/2+i}-1$ if and only if $i\equiv kt\mod 2t$; and
			$p^t+1\mid p^{n/2+i}+1$ if and only if $i\equiv (k+1)t\mod 2t$.
\end{enumerate}
\end{lemma}
\begin{proof}
   (i)
  First note that, since $p^t-1$ and $p^t+1$ have gcd equal to $1$ if $p=2$ and to $2$ if $p$ is odd,
  we have $\gcd(p^t+1,p^n-1)\gcd(p^t-1, p^n-1)=\delta_{t,n}\gcd((p^t+1)(p^t-1), p^n-1))$, where $\delta_{t,n}$ equals 1 if $p=2$ and is a power of 2 if $p$ is odd. It is a simple matter to see more precisely that $\delta_{t,n}$ equals the value defined above. Since
  $\gcd((p^t+1)(p^t-1), p^n-1))=\gcd(p^{2t}-1, p^n-1)= p^{\gcd(2t,4n)}-1$, this
  proves the value of $\gcd(p^t+1,p^n-1)$. The fact that $\delta_{t,n}\in \{1,2\}$ is obvious when $p=2$. For $p$ odd, from $\gcd(p^t+1,p^t-1)=2$, we
 have $\min(a,b)=1$, where $a,b$ are defined above.
 If $a=1$, then $\min(a,c)+\min(b,c)-\min(a+b,c)=1+\min(b,c)-\min(b+1,c)$ equals either $0$ or $1$,
 and therefore $\delta_{t,n}=1$ or $2$. Similarly we may show the case $b=1$.
\\ It is straightforward that, if $2t\mid n$ then the value obtained for $\gcd(p^t+1, p^n-1)$ equals $p^t+1$.
 Conversely:\\
 (1) If $\od(t) \geq \od(n)$ then $\gcd(2t,n)=\gcd(t,n)$ and hence
 $\gcd(p^t+1,p^{n}-1)=\delta_{t,n}$. Since $\delta_{t,n}\in\{1,2\}$
 and $p^t + 1 > 2$,  $p^t+1\mid p^n-1$ cannot happen. \\
 (2) If $\od(t) < \od(n)$ then
 $\gcd(p^t+1,p^{n}-1)=\delta_{t,n}(p^{\gcd(t,n)}+1)$. It is easy to check that in this case if $t\nmid n$, then
 $p^t+1\nmid p^n-1$ as $\delta_{t,n}(p^{\gcd(t,n)}+1)\neq p^t+1$, since $(p^t+1)/(p^{\gcd(t,n)}+1)\neq 1,2$. Therefore we have $t\mid n$.
 Further, one may check in this case we have $2t\mid n$.\\

 (ii) Similarly to the beginning of (i),
  since $p^n-1$ and $p^n+1$ have gcd equal to $1$ if $p=2$ and to $2$ if $p$ is odd,
  we have $\gcd(p^t+1,p^n+1)\gcd(p^t+1, p^n-1)=\eta_{t,n}'\gcd(p^t+1, (p^n-1)(p^n+1))$, where
  $\eta_{t,n}'=2^{\min(a,c)+\min(a,d)-\min(a,c+d)}$.
  Using (i), we get
  \begin{equation}
    \label{divide:eq1}
	\begin{array}{lll}
	  \gcd(p^t+1, p^n+1) &=& \eta_{t,n}'\frac{\gcd(p^t+1, p^{2n}-1)}{\gcd(p^t+1,p^n-1)} \\ [2ex]
	  &=& \eta_{t,n}'\left(\delta_{t,2n}\cdot\frac{p^{\gcd(2t,2n)}-1}{p^{\gcd(t,2n)}-1}\right)
	  \left(\frac{1}{\delta_{t,n}}\cdot\frac{p^{\gcd(t,n)}-1}{p^{\gcd(2t,n)-1}}\right) \\ [2ex]
	   &=& \frac{\eta_{t,n}'\delta_{t,2n}}{\delta_{t,n}}\cdot
	   \frac{(p^{\gcd(2t,2n)}-1)(p^{\gcd(t,n)}-1)}{(p^{\gcd(t,2n)}-1)(p^{\gcd(2t,n)}-1)}.
	\end{array}
  \end{equation}
  This completes the proof of the value of $\gcd(p^t+1,p^n+1)$.
  It is tedious but easy to show that $\eta_{t,n}\in\{1,2\}$.\\
  (1) If $\od(t)= \od(n)$, then from (\ref{divide:eq1}) we have
  \begin{eqnarray*}
    \gcd(p^t+1,p^n+1) &=&  \eta_{t,n}\cdot \left(p^{\gcd(t,n)}+1\right).
  \end{eqnarray*}
  It is easy to see that if $t\nmid n$ then $p^t+1\nmid p^n+1$ since $(p^t+1)/(p^{\gcd(t,n)}+1)\neq 1,2$.
  Hence we have $t\mid n$ and furthermore in this case $n/t$ is odd.\\
  (2) If $\od(t)\neq \od(n)$, then one may easily verify that in this case $\gcd(p^t+1, p^n+1)=\eta_{t,n}$.
  Since $\eta_{t,n}\in\{1,2\}$ and $p^t+1>2$, in this case $p^t+1\nmid p^n+1$.\\
  Conversely, if $t\mid n$ and $n/t$ is odd, one may easily verify that $p^t+1\mid p^n+1$.

(iii)
  By (i), we have $p^t+1\mid p^{n/2+i}-1$ if and only if $2t \mid n/2+i$,
  i.e.  $i\equiv -kt\mod 2t\equiv kt \mod 2t$ we proved then the first part.
  Next, by (ii), we have $p^t+1\mid p^{n/2+i}+1$ if and only if
  $n/2+i=kt+i=\ell t$ for some odd integer $\ell$. Equivalently,
  $i\equiv (k+1)t\mod 2t$. This completes the proof.
\end{proof}

\begin{lemma}
  \label{formofF}
  Let $F:\gf_{p^n}\mapsto \gf_{p^n}$ be defined as
  $F(x)=G(x^{p^t+1})$,  for some function $G$ over $\gf_{p^n}$.
  If  $F$ is quadratic, then
  $F$ can be represented in the form $F(x)=\sum\limits_{i\in \Omega}a_{i}x^{p^{k_i}+p^{\ell_i}}$,
  where $\Omega$ is some subset of $\mathbb{N}$ and, for every $i\in\Omega$, $t\mid (k_i-\ell_i)$ and $(k_i-\ell_i)/t$ is odd.
\end{lemma}
\begin{proof}
  Since $F(x)=G(x^{p^t+1})$,
  then we have $F(x)=\sum_{j\ge 0}b_jx^{j(p^t+1)}$, where $b_j\in\gf_{p^n}$ (note that we do not bound here the values of the exponents since we do not reduce modulo $x^{p^n}-x$). Since $F$ is quadratic, we have $\sum_{j\ge 0}b_jx^{j(p^t+1)}=\sum_{i\in\Omega}a_ix^{p^{k_i}+p^{\ell_i}}$,
  where $a_i\in\gf_{p^n}$ and $\Omega$ is a subset of $\mathbb{N}$ (note indeed that the form $p^{k_i}+p^{\ell_i}$ of the exponents is invariant under congruence modulo $p^n-1$). We can assume that $p^{k_i}\ge p^{\ell_i}$ for each $i$ in $\Omega$. For each $i$ in $\Omega$, we have $p^t+1\mid p^{k_i}+p^{\ell_i}=p^{\ell_i}(p^{k_i-\ell_i}+1)$. Hence
  $p^t+1\mid p^{k_i-\ell_i}+1$. By Lemma \ref{divide}(ii), we obtain
  $t\mid k_i-\ell_i$ and $(k_i-\ell_i)/t$ is odd. The proof is completed.
\end{proof}

\begin{lemma}
  \label{dimEa}
  Let $F(x)=G(x^d)$ be a quadratic function from $\gf_{p^n}$ to itself, where $p$ is any prime,  $n=2kt$ and $\gcd(d,p^n-1)=p^t+1$ for some non-negative integer $t$.
  For each nonzero $a\in\gf_{p^n}$, define
  $$ E_a=\{s : s\in\gf_{p^n}| \tr(a(F(y+s)-F(y)))=0\ \mbox{for all}\ y\in\gf_{p^n} \}. $$
  Then $E_a$ is a vector space over $\gf_p$ with even dimension.
\end{lemma}
\begin{proof}
  Clearly $E_a\neq \emptyset$ as $0\in E_a$.
  Moreover, $E_a$ is a vector space over $\gf_p$ since if $\tr(a(F(y+s)-F(y)))$
  and $\tr(a(F(y+s')-F(y)))$ both equal the null function, then for every
  $u,v\in \gf_p$, $\tr(a(F(y+us)-F(y)))$ and $\tr(a(F(y-vs')-F(y)))$
  are also null and by subtraction $\tr(a(F(y+us)-F(y-vs')))$ is null and
  then $\tr(a(F(y+us+vs')-F(y)))$  is the null function.

  When $p=2$, it is well-known (see e.g. \cite{Cbook1}) that the dimension of $E_a$ is even.
  When $p$ is an odd prime, let us show that
  $E_a$ is a vector space over $\gf_{p^2}$, and therefore,
  the dimension of $E_a$ as a vector space over $\gf_p$ is even.

  By Lemma \ref{formofF}, $F$ has the form
  $F(x)=\sum\limits_{i\in\Omega}a_{i}x^{p^{k_i}+p^{\ell_i}}$,
  where $t\mid k_i-\ell_i$ and $(k_i-\ell_i)/t$ is odd for each $i\in\Omega$.
  An element  $s$ of $\gf_{p^n}$ belongs to $E_a$ if and only if:
  \begin{equation}
    \label{allzero}
 \forall y\in \gf_{p^n},\,    \tr(a(F(y+s)-F(y)))=0.
  \end{equation}
  Substituting $F(x)=\sum\limits_{i\in\Omega}a_{i}x^{p^{k_i}+p^{\ell_i}}$ into (\ref{allzero}) and simplifying, we have then:
  \begin{equation*}
   \forall y\in \gf_{p^n},\,    \tr\left(a\sum\limits_ia_i(s^{p^{k_i}+p^{\ell_i}}+s^{p^{k_i}}y^{p^{\ell_i}}
     +s^{p^{\ell_i}}y^{p^{k_i}}   ) \right)=0.
  \end{equation*}
  Therefore we have $s\in E_a$ if and only if, for every $y\in \gf_{p^n}$:
  \begin{equation}
    \label{lemma8:eq1}
    \begin{array}{lll}
	 -\tr\left(\sum\limits_iaa_is^{p^{k_i}+p^{\ell_i}}\right) &=&
	\tr\left(\sum\limits_iaa_i(s^{p^{k_i}}y^{p^{\ell_i}}
	 +s^{p^{\ell_i}}y^{p^{k_i}}   ) \right)   \\
	 &=& \tr\left(\sum\limits_i
	 \left((aa_is^{p^{k_i}})^{p^{-\ell_i}} + (aa_is^{p^{\ell_i}})^{p^{-k_i}}\right)y
	 \right).
     \end{array}
   \end{equation}
   Note that the left hand side of  (\ref{lemma8:eq1}) equals $-\tr(aF(s))$ and is null
 by letting $y=0$ in (\ref{allzero}) (we may assume without loss of generality that
   $F(0)=0$ as otherwise we may replace $F(x)$ with $F'(x)=F(x)-F(a)$ and $F'$ satisfies all properties
   of $F$ which are stated in the hypothesis). The condition becomes:
   \begin{equation}
     \label{lemma8:eq2}
	\sum\limits_i
	 \left((aa_is^{p^{k_i}})^{p^{-\ell_i}} + (aa_is^{p^{\ell_i}})^{p^{-k_i}}\right)=0.
   \end{equation}
  Recall that we have $t\mid (k_i-\ell_i)$ and $(k_i-\ell_i)/t$ is odd.
  In the following let $k_i-\ell_i=e_it$ for $i$ in $\Omega$, where $e_i$
  is some odd integer. Then (\ref{lemma8:eq2}) becomes:
  \begin{equation}
     \label{lemma8:eq3}
	\sum\limits_{i=0}^{n-1}
	\left((aa_i)^{p^{-\ell_i}}s^{p^{e_it}} + (aa_i)^{p^{-k_i}}s^{p^{-e_it}}\right)=0.
  \end{equation}
For any $s\in E_a$ and any $w\in\gf_{p^2}$, we have $ws\in E_a$. Indeed, $w^{p^e}$ equals $w$ if $e$ is even and equals $w^p$ if $e$ is odd. Then $(aa_i)^{p^{-\ell_i}}(ws)^{p^{e_it}} + (aa_i)^{p^{-k_i}}(ws)^{p^{-e_it}}=w^{p^t}[(aa_i)^{p^{-\ell_i}}s^{p^{e_it}} + (aa_i)^{p^{-k_i}}s^{p^{-e_it}}]=0$.
  For any constants $w_1,w_2\in\gf_{p^2}$ and any $s_1, s_2\in E_a$,
  we have then $w_1s_1+w_2s_2\in E_a$, since
  $\tr(a(F(y+w_1s_1)-F(y)))$ and $\tr(a(F(y-w_2s_2)-F(y)))$
  are constant zero and by subtraction $\tr(a(F(y+w_1s_1)-F(y-w_2s_2)))$ is constant zero and then $\tr(a(F(y+w_1s_1+w_2s_2)-F(y)))$  is constant zero.
\end{proof}

We will need the following result in \cite[Lemma 3]{Feng-Luo} for the proof of the main theorem.
\begin{lemma}
\label{Feng-Luo}
Let $f$ be a quadratic perfect nonlinear function with $f(-x)=f(x)$ for all nonzero $x\in\gf_{p^n}$ and $f(0)=0$.
Then the Walsh coefficient $\mathcal{W}_{\tr(aF)}(0)=\epsilon_{a,0}(\sqrt{p^*})^n$,
where $\epsilon_{a,0}\in\{\pm1\}$, $p^*=\left(\frac{-1}{p}\right)p$ and $\left(\frac{-1}{p}\right)$ is the Legendre symbol.
\end{lemma}

Now we are ready to give the main result of this Section. Note that the first part of the
following result first appeared in \cite{weng}. We include it here and give a short proof
for the completeness of the paper. Note that the second part of the following result
may give rise to new SRGs by using certain ZDB functions and comparing them with the 
known constructions in \cite{databasesrg,Xiang-srg,ma-pds}.

\begin{theorem}
  \label{gengraph}
  Let $F(x)=G(x^d)$ be a quadratic function from $\gf_{p^n}$ to itself, where $p$ is any prime and
  $\gcd(d,p^n-1)=p^t+1$ for some non-negative integer $t$.
  Assume that the restriction of $G$ to $C_d=\{x^d : x\in \gf_{p^n}^*\}=C_{p^t+1}$ is an injection from $C_d$ to $\gf_{p^n}$.
  Define the set $D=\{ F(x) : x\in\gf_{p^n}  \}\setminus\{0\}$. Then:
  \begin{enumerate}[(i)]
    \item if $t=0$ and $p$ is an odd prime, then $D$ is a
	    \begin{eqnarray*}
		    & \left( p^n, \frac{p^n-1}{2}, \frac{p^n-3}{4} \right) \ \mbox{difference set}, &\mbox{when}\ p^n\equiv 3\mod 4, \\
		    & \left( p^n, \frac{p^n-1}{2}, \frac{p^n-5}{4}, \frac{p^n-1}{4} \right) \ \mbox{partial difference set},
		    &\mbox{when}\ p^n\equiv 1\mod 4.
	    \end{eqnarray*}
    \item if $t>0$ and $n$ is divisible by $2t$, then $D$ is a
	    \begin{equation*}
		    \left(p^n, \frac{p^n-1}{p^t+1},
		          \frac{p^n-3p^t-2-\epsilon p^{n/2+2t}+\epsilon p^{n/2+t}}{(p^t+1)^2},
                          \frac{p^n-\epsilon p^{n/2}+\epsilon p^{n/2+t}-p^t}{(p^t+1)^2}
            \right)
	    \end{equation*}
	    partial difference set, where $n=2kt$ and $\epsilon=(-1)^{k}$.
  \end{enumerate}
\end{theorem}
\begin{proof}
	Without loss of generality, we may assume that $d=p^t+1$.
	Let us denote the additive group of $\gf_{p^n}$ by ${\cal G}$.
	By Corollary \ref{invfor}, to prove that $D$ is a (partial) difference set with the prescribed parameters,
	we need to determine the character values of $D$. Now, for each nontrivial character
	$\chi_a\in\widehat{{\cal G}}, a\in {\cal G}^*$, we have
	$\chi_a(D)=\sum\limits_{x\in C_d}\zeta_p^{\tr(aG(x))}$, where $\zeta_p$ is the chosen $p$-th root of unity.
	It is not difficult to see that
	\begin{eqnarray*}
		\mathcal{W}_{\tr(aF)}(0) = 
		\sum\limits_{x\in\gf_{p^n}}\zeta_p^{\tr(aF(x))}=1+d\sum\limits_{x\in C_d}\zeta_p^{\tr(aG(x))}
				    = 1+d\chi_a(D),
	\end{eqnarray*}
	and hence
	\begin{equation}
		\label{chiD}
		\chi_a(D)=\frac{1}{d}\left( \mathcal{W}_{\tr(aF)}(0) - 1 \right).
	\end{equation}
	Denoting $\mathcal{W}_{\tr(aF)}(0)$ by $X_a$, we have
	\begin{eqnarray}
		\label{xxbar}
		X_a\overline{X_a}=\sum\limits_{x,y\in\gf_{p^n}}\zeta_p^{\tr(a(F(x)-F(y)))}     
		=\sum\limits_{t\in\gf_{p^n}}\left(\sum\limits_{y\in\gf_{p^n}}\zeta_p^{\tr(a(F(y+t)-F(y)))}\right).
	\end{eqnarray}
	
	(i): For $t=0$ we have $d=p^0+1=2$. By hypothesis, we assume that $F(x)=G(x^2)$ is a quadratic function and that 
	$G|_{C_2}$ is an injection.
	Then $F$ is a PN function (see \cite{weng} or Corollary \ref{dto1map} above)
	and then $F(y+t)-F(y)$ is a PP
	over $\gf_{p^n}$ for any nonzero $t$. Therefore, by (\ref{xxbar}) we have 	
	$X_a\overline{X_a}=p^n$. By Lemma \ref{Feng-Luo}, we have $X_a=\varsigma(\sqrt{p^*})^n$, 
	where $\varsigma\in\{-1,1\}$ and $p^*=\left(\frac{-1}{p}\right)p$.
	In the following we divide the proof into two cases.
		
	Case 1: $n$ is even. Note that in this case $p^n\equiv 1\mod 4$, 
	which implies that  
	$X_a=\varsigma(\sqrt{p^*})^n=\varsigma\left((\frac{-1}{p})p\right)^{n/2}=\varsigma (\frac{-1}{p})^{n/2}p^{n/2}$.
	Hence we have $\chi_a(D)=\frac{1}{2}(\varsigma (\frac{-1}{p})^{n/2}p^{n/2}-1)$. It can be verified that 
	$\chi_a(DD^{(-1)})=\chi_a(D)\chi_a(D^{(-1)})=\chi_a(D)\overline{\chi_a(D)}
	=\frac{1}{4}(p^n+1-2\varsigma(\frac{-1}{p})^{n/2} p^{n/2})$. On the other hand,
	it can be easily computed that $(k-\lambda)+(\mu-\lambda)\chi_a(D)=\frac{1}{4}(p^n+1-2\varsigma(\frac{-1}{p})^{n/2} p^{n/2}))$. 
	Then, by Lemma \ref{pdsiff}, we have that $D$ is a PDS with the prescribed parameter.
			
	Case 2: $p\equiv 3\mod 4$ and $n$ is odd. Assume that $n=2m+1$. In this case we have 
	$X_a=\varsigma(\sqrt{p^*})^{2m+1}=\varsigma\left((\frac{-1}{p})p\right)^{m}\sqrt{p^*}=
	\varsigma\left(\frac{-1}{p}\right)^mp^m\sqrt{p^*}$, then $\chi_a(D)=\frac{1}{2}(\varsigma\left(\frac{-1}{p}\right)^mp^m\sqrt{p^*}-1)$.
	On the one hand, note that the complex conjugate of $\sqrt{p^*}$ equals $-\sqrt{p^*}$ 
	(since $\sqrt{p^*}\cdot (-\sqrt{p^*})=-p^*=-\left(\frac{-1}{p}\right)p=p=|\sqrt{p^*}|^2$). 
	Then $\chi(DD^{(-1)})=\chi(D)\overline{\chi(D)}=\frac{1}{4}(\varsigma\left(\frac{-1}{p}\right)^mp^m\sqrt{p^*}-1)(\varsigma\left(\frac{-1}{p}\right)^mp^m\overline{\sqrt{p^*}}-1)=\frac{1}{4}(\varsigma\left(\frac{-1}{p}\right)^mp^m\sqrt{p^*}-1)(-\varsigma\left(\frac{-1}{p}\right)^mp^m\sqrt{p^*}-1)=-\frac{1}{4}(\left(\frac{-1}{p}\right)p^n-1)=\frac{1}{4}(p^n+1)$ (since $\left(\frac{-1}{p}\right)=-1$ as $p\equiv3\mod 4$). On the other hand, one may compute that $k-\lambda=\frac{1}{4}(p^n+1)$. By Lemma \ref{pdsiff}, we prove that $D$ is the difference set with the prescribed parameters. \\
			
			(ii) Next, we deal with the case that $d=p^t+1$ with $t>0$. First we claim that $D=D^{(-1)}$ in this case. When $p=2$, this is clear. When $p$ is an odd prime, we have $p^n\equiv 1\mod 4$ since $n$ is even and then $-1$ is a square in $\gf_{p^n}$. Since $2(p^t+1)\mid p^n-1$ (one may see that $p^n-1=p^{2kt}-1=(p^t+1)(1-p^t+p^{2t}+\cdots+(-1)^{(2k-1)}p^{(2k-1)t})$ and the second factor is even), there exists an element $\alpha\in\gf_{p^n}$ such that its order is $2(p^t+1)$, namely $\alpha^{p^t+1}=-1$. By Lemma \ref{formofF}, we have $F(\alpha x)=\sum_{i\in\Omega}a_i(\alpha x)^{p^{k_i}+p^{\ell_i}}=\sum_{i\in\Omega}a_i(\alpha x)^{p^{\ell_i}(p^{e_it}+1)}$, where $e_i=(k_i-\ell_i)/t$ is odd by Lemma \ref{formofF}. Since $\alpha^{p^{\ell_i}(p^{e_it}+1)}=\alpha^{p^{\ell_i}(p^t+1)(1-p^t+\cdots+(-1)^{e_i-1}p^{(e_i-1)t})}=\left(\alpha^{p^t+1}\right)^{p^{\ell_i}(1-p^t+\cdots+(-1)^{e_i-1}p^{(e_i-1)t})}=-1$. Therefore from the above computations we get $F(\alpha x)=-F(x)$ for every $x\in\gf_{p^n}$, which implies that $D=D^{(-1)}$.
	
Further, since $F$ is quadratic, the function $f_{a,s}(y)=\tr(a(F(y+s)-F(y)))$ is affine for every $a,s\in\gf_{p^n}$ and the sum $\sum_{y\in\gf_{p^n}}\zeta_p^{\tr(a(F(y+s)-F(y)))}$ is then nonzero only when $f_{a,s}(y)$ is a constant function. For each $a\in {\cal G}^*$, let us define then
    $$
      E_a=\{s : s\in\gf_{p^n}| \tr(a(F(y+s)-F(y)))\ \mbox{is a constant for all}\ y\in\gf_{p^n} \}.
    $$
	According to Corollary \ref{dto1map},  for each nonzero $s$, there always exists $y\in\gf_{p^n}$ such that $F(y+s)-F(y)=0$.
	As a result, the set $E_a$ actually is
	$$ E_a=\{s:s\in\gf_{p^n}| \tr(a(F(y+s)-F(y)))=0\ \mbox{for all}\ y\in\gf_{p^n} \} $$
	and Relation (\ref{xxbar}) becomes $
		X_a\overline{X_a}=p^n|E_a|$.
\\
     By Lemma \ref{dimEa}, we know that $E_a$ is a vector space over $\gf_p$ and its dimension is even.
	For all $0\leq i\le n/2$, let $N_i$ be the number of nonzero $a\in\gf_{p^n}$ such that $|E_a|=p^{2i}$.
	Now we consider the following sum,
	\begin{equation}
		\label{threesum}
		\sum\limits_{a\in\gf_{p^n}^*}X_a\overline{X_a}=
	        \sum\limits_{a\in\gf_{p^n}^*}\sum\limits_{x,y\in\gf_{p^n}}\zeta_p^{\tr(a(F(x)-F(y)))}.
        \end{equation}
	On the one hand, the LHS of (\ref{threesum}) equals
	$\sum_{a\in\gf_{p^n}^*}p^n|E_a|=\sum_{i\geq 0}p^{n+2i}N_i$.
	On the other hand, the RHS of (\ref{threesum}) is
	\begin{eqnarray*}
	    \mbox{RHS} &=& \sum_{x,y\in\gf_{p^n}}\sum_{a\in\gf_{p^n}}\zeta_p^{\tr(a(F(x)-F(y)))}-p^{2n}
	    = p^n|\{(x,y)\in\gf_{p^n}^2 | F(x)=F(y) \}| - p^{2n} 	   \\
	    &=& p^n(1+d(p^n-1))-p^{2n} = (d-1)p^n(p^n-1)=(d-1)p^n\sum\limits_{i\geq 0}N_i.
        \end{eqnarray*}
	Hence, we have $\sum_{i\geq 0}p^{n+2i}N_i=(d-1)p^n\sum_{i\geq 0}N_i$, which implies that
	$\sum_{i\geq 0}(p^{2i}-(d-1))N_i=\sum_{i\geq 0}(p^{2i}-p^t)N_i=0$. Rearrange the term,
	we get
	\begin{equation}
		\label{N0eq1}
		(p^t-1)N_0=\sum\limits_{i\geq 1}(p^{2i}-p^t)N_i.
	\end{equation}
	Since $D=D^{(-1)}$ and $\chi_a(D)\in\mathbb{Z}[\zeta_p]$, we have $\chi_a(D)\in\mathbb{Z}$ by Lemma \ref{rationalinteger} and then
    $X_a=d\chi_a(D)+1\in\mathbb{Z}$.
	Therefore, from $X_a\overline{X_a}=p^{n+2i}$ for some $i\geq 0$ we get
	$X_a=\pm p^{n/2+i}$ since $n$ is even.
    Furthermore, by Lemma \ref{divide}(iii) and $\chi_a(D)=\frac{1}{p^t+1}(X_a-1)=\frac{1}{p^t+1}(\pm p^{n/2+i}-1)\in\mathbb{Z}$,
    we have that $N_i\neq 0$ if and only if $i\mod 2t\in\{ kt, (k+1)t \}$, and
    \begin{equation}
    \label{chiDvalue}
     (X_a,\chi_a(D))=\left\{
    \begin{array}{ll}
      \left(p^{n/2+i}, \frac{p^{n/2+i}-1}{p^t+1}\right)  &\mbox{if}\ i\equiv kt\mod 2t,  \\
      \left(-p^{n/2+i}, \frac{-p^{n/2+i}-1}{p^t+1}\right)  &\mbox{if}\ i\equiv (k+1)t\mod 2t. \\
    \end{array}
  \right.
  \end{equation}
  Let $\Phi$ be the set $\{0\leq i\leq n/2 | i\mod 2t\in\{kt, (k+1)t\}\}$.
	Hence, since
	$\sum_{a\in\gf_{p^n}^*}X_a=\sum_{a\in \gf_{p^n}^*}\sum_{x\in \gf_{p^n}}\zeta_p^{\tr(aF(x))}=p^n|\{x\in \gf_{p^n}\, |\, F(x)=0\}|-p^n=0$
    we have
    $\sum_{i\in\Phi}(-p)^{(n/2+i)/t}N_i= 0,$ that is:
    \begin{equation}
    \label{thm2:eq5}
    \sum_{i\in\Phi\atop i\equiv kt\mod 2t}p^iN_i=\sum_{i\in\Phi\atop i\equiv (k+1)t\mod 2t}p^iN_i.
    \end{equation}

    By equation (\ref{thm2:eq5}) we have $\sum_{i\in\Phi,i\equiv 0\mod 2t}p^iN_i=\sum_{i\in\Phi,i\equiv t\mod 2t}p^iN_i$
    and then $N_0=\sum_{i\in\Phi,i\neq 0}(-1)^{i/t-1}p^iN_i$. By equation (\ref{N0eq1}) we have $N_0=\sum_{i\in\Phi,i\neq 0}\frac{p^{2i}-p^t}{p^t-1}N_i$. Note that one may easily check that $\frac{p^{2i}-p^t}{p^t-1}\ge p^i$ when $i\ge t$ and the equality holds if and only if $i=t$. Now, if $N_i$ is nonzero for any $i\in\Phi\setminus\{0,t\}$, we will have
    $$ 
     N_0=\sum_{i\in\Phi,i\neq 0}\frac{p^{2i}-p^t}{p^t-1}N_i > \sum_{i\in\Phi,i\neq 0}(-1)^{i/t-1}p^iN_i=N_0,
    $$
    which is a contradiction.
    Therefore we have $|E_a|=1$ or $p^{2t}$ for any $a\in\gf_{p^n}^*$, and then the values of
	$X_a$ lie in the set $\left\{\pm p^{n/2}, \pm p^{n/2+t} \right\}$.
    We divide the following proof into two cases according to the parity of $k$.
    
    Case 1: $k$ is even. By (\ref{chiDvalue}), we have $X_a\in\{p^{n/2}, -p^{n/2+t}\}$
    and $\chi_a(D)\in\{\frac{p^{n/2}-1}{p^t+1}, \frac{-p^{n/2+t}-1}{p^t+1}\}$.

    Case 2: $k$ is odd. Again, by (\ref{chiDvalue}), we have $X_a\in\{-p^{n/2}, p^{n/2+t}$
    and $\chi_a(D)\in\{-\frac{p^{n/2}-1}{p^t+1}, \frac{p^{n/2+t}-1}{p^t+1}\}$.
    
    Finally, by Lemma \ref{pdsiff}(ii) and the fact that $D=D^{(-1)}$, we have that $D$ is a PDS
    with the prescribed parameters. We complete the proof.
\end{proof}

Particularly, when $p=2$ and $t=1$, we may obtain PDS from APN functions.
\begin{corollary}
\label{APNgraph}
Let $F$ be a quadratic APN function on $\gf_{2^{n}}$ with the form $F(x)=G(x^3)$, where $G|_{C_3}$
is an injection and $n=2k$. Let $D$ denote the set
$$
D=\{ F(x) : x\in\gf_{2^n} \} \setminus \{ 0 \}.
$$
Then $D$ is a partial difference set with parameters
\begin{eqnarray*}
&\left(2^n, \frac{2^n-1}{3}, \frac{1}{9}(2^k+4)(2^k-2),
\frac{1}{9}(2^k+1)(2^k-2) \right) & \ \mbox{if}\ k\ \mbox{is odd}, \\
[2ex] &\left(2^n, \frac{2^n-1}{3}, \frac{1}{9}(2^k-4)(2^k+2),
\frac{1}{9}(2^k-1)(2^k+2) \right) & \ \mbox{if}\ k\ \mbox{is even}. \\
\end{eqnarray*}
\end{corollary}

From the proof of Theorem \ref{gengraph}, we have the following result, which may be
used to determine the Walsh spectrum of $F$ in Theorem \ref{gengraph}.

\begin{corollary}
   \label{hyperplane}
  Let $F(x)=G(x^d)$ be a quadratic function from $\gf_{p^n}$ to itself, where $p$ is a prime,
  $\gcd(d,p^n-1)=p^t+1$ for some integer $t>0$ and $n=2kt$ for 
  some positive integer $k$.
  Assume that the restriction of $F$ to $C_d$ is an injection from $C_d$ to $\gf_{p^n}$.
  For any $a\in\gf_{p^n}^*$, define the set
   $$ E_a = \{ s : s\in\gf_{p^n} | \tr\left( a(F(y+s)-F(y)) \right)=0\ \mbox{for all}\ y\in\gf_{p^n} \}. $$
   Then $|E_a|$ is either $1$ or $p^{2t}$.
\end{corollary}

\begin{theorem}
\label{walsh}
    Let $F(x)=G(x^d)$ be a quadratic function from $\gf_{p^n}$ to itself, where $p$ is a prime,
    $\gcd(d,p^n-1)=p^t+1$ for some integer $t>0$ and  $n=2kt$ for
  some positive integer $k$.
    Then, for any $a,b\in\gf_{p^n}$, the Walsh coefficient $\mathcal{W}_F(a,b)$ satisfies
    $$ |\mathcal{W}_{F}(a, b)|^2\in\{0, p^n, p^{n+2t} \}. $$
    Particularly, if $n$ is even and $F=G(x^3)$ is a quadratic APN function on $\gf_{2^n}$,
    then the Walsh spectrum of $F$ is $\{0, \pm 2^{n/2}, \pm 2^{n/2+1} \}$.
\end{theorem}
\begin{proof}
	For any $a, b\in\gf_{2^n}$, we have
	\begin{equation}
		\begin{array}{lll}
		\mathcal{W}_{F}(a, b)\overline{\mathcal{W}_{F}(a,b)} &=&
                             \sum\limits_{x,y\in\gf_{p^n}}\zeta_p^{\tr(a(F(x)-F(y))+b(x-y))}
                           \\
                           &=& \sum\limits_{t, x}\zeta_p^{\tr(aF(x+t)-aF(x)+bt)} \\
                           &=& \sum\limits_t\zeta_p^{\tr(bt)}\sum\limits_x\zeta_p^{\tr(a(F(x+t)-F(x)))}
                           \\
                           &=&p^n \sum\limits_{t\in E_a}\zeta_p^{\tr(bt)}=\chi_b(E_a),     \\
		\end{array}
	\end{equation}
    where $\chi_b$ is the character of $\gf_{p^n}$.
    It is then clear that $\mathcal{W}_{F}(a, b)\overline{\mathcal{W}_{F}(a,b)}$ lies in the set
    $\{0, p^n|E_a|\}$ since $E_a$ is an affine subspace. By Corollary \ref{hyperplane}, we have
    $\mathcal{W}_{F}(a, b)\overline{\mathcal{W}_{F}(a,b)}\in\{0,p^n,p^{n+2t}\}$.
    Particularly, when $F=G(x^3)$ is a quadratic APN function on $\gf_{2^n}$ with $n$ even,
    we have $\mathcal{W}_{F}(a, b)\in\mathbb{Z}$ and hence
    $\mathcal{W}_{F}(a, b)\in\{0,\pm p^{n/2}, \pm p^{n/2+1}\}$.
    The proof is completed.
\end{proof}

\begin{remark}
The above result shows that the Walsh spectrum of the APN function with the form $F(x)=G(x^3)$
is the same as the one of Gold APN function. Theorem \ref{walsh} presents a unifying treatment of
determining the Walsh spectrum of quadratic APN functions with the form $F(x)=G(x^3)$, where $G$ is
an injection from $C_3$ to $\gf_{2^n}$. For instance, it includes the
result in \cite{Bracken-APN} when $n$ is even.
\end{remark}

\section{Newness}
In this section, we discuss the newness of the SRGs generated by the quadratic zero-difference $p^t$-balanced 
functions in Theorem \ref{gengraph}. We first give some general results,
and then we show that some of the obtained negative Latin square type SRGs are new by comparing
them to the known constructions.

\subsection{Isomorphism of graphs and equivalence of functions}
We first introduce some definitions. Two graphs $\mathtt{G}_1=(V_1,E_1)$ and $\mathtt{G}_2=(V_2,E_2)$
are called \textit{isomorphic} if there exists a one-to-one function $\sigma$ mapping $V_1$ to $V_2$,
and $E_1$ to $E_2$ such that for each
pair $(P,e)\in V_1\times E_1$, we have $\sigma(P)\in \sigma(e)$ if and only if
$P\in e$. Let $D_1,D_2$ be two partial difference sets of the group ${\cal G}$, they 
are called \textit{CI}-equivalent if there exists an automorphism $\phi\in\mbox{Aut}({\cal G})$
such that $\phi(D_1)=D_2$. Moreover, they are said to be \textit{SRG}-equivalent if the corresponding
Cayley graphs are isomorphic, i.e. $\mbox{Cay}({\cal G}, D_1)\cong \mbox{Cay}({\cal G}, D_2)$.
It is known that SRG-equivalence implies CI-equivalence, but not vice versa (see an example in \cite{counterexample}).
For more details on CI-equivalence and SRG-equivalence, one may refer to \cite{weng}.

It is not difficult to see that the property of zero-difference balancedness for a function $F$ is not necessarily preserved
under EA-equivalence. However, the zero-difference balancedness property is 
preserved by affine equivalence and induces isomorphism of the associated Cayley graphs, but not the addition of a linear function. 
It is worth pointing out that EA-inequivalent functions may not necessarily lead to
non-isomorphic graphs (see the following Table 1). Indeed, $18$ EA-inequivalent APN functions
with the form $G(x^3)$ on $\gf_{2^8}$ were found in \cite{tan-APN} (independently found in \cite{yu-APN}),  
where $G$ is an injection on the set of nonzero cubes (listed in the Appendix),
the No. 2 and 6 functions, and the No. 13, 14 and 17 functions lead to isomorphic SRGs
(see Table 1).

\subsection{New negative Latin square type SRGs}

It is well known that there are many non-isomorphic Latin square type SRGs by the
following construction, therefore we only focus on the newness of negative Latin
square type SRGs. Recall that the definition of (negative) Latin square type SRGs
can be found at the end of Section 2.3.

\begin{lemma}[PCP construction, \cite{ma-pds}]
	\label{PCP}
Let ${\cal G}$ be the additive group of a $2$-dimensional vector space $V$ over $\gf_q$.
Let $H_1, H_2, \ldots, H_r$, where $r\leq q+1$, be $r$ hyperplanes of $V$.
Then $D=(H_1\cup H_2\cup \cdots \cup H_r)\setminus \{0\}$ is a
$(q^2, r(q-1), q + r^2-3r, r^2-r)$ partial difference set in ${\cal G}$.
\end{lemma}

In the following we list the known constructions which may generate negative Latin square type
SRGs with the same parameters in Theorem \ref{gengraph}(ii). Recall that the $e$-th cyclotomic
classes of a finite field $\gf_{p^n}$ ($p^n=ef+1$) are the subsets 
$$ C_j^e=\{w^{ie+j}:i=0,\cdots,f-1\}, \ \  (0\leq j\leq e-1), $$
where $w$ is a primitive element of $\gf_{p^n}$.

\begin{lemma}[Calderbank and Kantor, \cite{cyclosrg}]
	\label{cyclosrg}
	Let $q$ be a prime power and $C_0,C_1,\ldots,C_q$ be the $(q+1)$-th cyclotomic classes in
	$\gf_{q^{2m}}$. For any $I\subset\{0, 1, \ldots, q\}$, $D=\cup_{i\in I}C_i$ is a
	regular $(q^{2m}, u(q^{2m}-1)/(q+1), u^2\eta^2+(3u-q-1)\eta-1, u^2\eta^2+u\eta)$-PDS
	in the additive group of $\gf_{q^{2m}}$ where $u=|I|$ and $\eta=( (-q)^m -1)/(q+1)$.
\end{lemma}

In \cite[Theorem 2]{Xiang-srg}, Brouwer, Wilson and Xiang generalized the construction of SRGs in
the above Lemma \ref{cyclosrg}. Their construction requires the so-called semiprimitive condition.
We shall elaborate below (before Table 1) that our constructions of SRGs in Theorem 
\ref{gengraph}(ii) is more general since it does not require the semiprimitive condition
and it may generate new SRGs which are not covered by \cite[Theorem 2]{Xiang-srg}.

Finally, it is well known that negative Latin square type SRGs may be obtained from
projective two-weight codes (see definition in \cite{cyclosrg}) as follows:
\begin{lemma}
    [\cite{ma-pds}]
	\label{twoweightcode}
	Let $y_1, y_2, \ldots, y_n$ be pairwise linear independent vectors in $\gf_q^n$. Then $y_1, y_2, \ldots, y_n$
	span a two-weight $[n, s]$-projective code $\mathcal{C}$ if and only if
	$$ D=\{ty_i : t\in\gf_q\setminus\{0\} \ \mbox{and}\ i=1,2,\ldots, n \} $$
	is a regular PDS in the additive group of $\gf_q^s$. Furthermore, if the two nonzero weights of $\mathcal{C}$
	are $w_1$ and $w_2$, then $D$ is a
	$$
	  (q^s, n(q-1), k^2+3k-q(k+1)(w_1+w_2)+q^2w_1w_2, k^2+k-qk(w_1+w_2)+q^2w_1w_2)
    $$
    partial difference set.
\end{lemma}

In the following, with the help of a computer, we show that, using the construction of SRGs in Theorem \ref{gengraph}(ii)
and the aforementioned 18 APN functions on $\gf_{2^8}$, new negative Latin square type SRGs are obtained. 

In Theorem \ref{gengraph}, by letting $p=2, t=1, n=8$, the $18$ APN functions of the form
$G(x^3)$ where $G|_{C_3}$ is an injection lead to $(256,85,24,30)$-SRGs. By the 
online database of known constructions of SRGs \cite{databasesrg}, SRGs with these parameters can be
constructed from the following methods:
\begin{enumerate}[(i)]
   \item The SRG from Lemma \ref{cyclosrg} by letting $u=1$. This graph is verified to be isomorphic to
	   the SRGs with No. 13, 14, 17 APN functions in Appendix.
   \item The SRG from projective binary $[85,8]$ two-weight codes with weights $40,48$ as in Lemma \ref{twoweightcode}.
           Checking the online database of two-weight codes \cite{databasetwoweightcode} shows that there are three constructions of
           such codes \cite{858code}. By a computer the SRGs from these codes are all isomorphic to the one from Lemma \ref{cyclosrg} above
	       (actually by Magma these three codes are equivalent).
   \item In \cite[Theorem 2]{Xiang-srg}, to obtain an SRG with the above parameter, we need to require
         (using the same notation as in \cite[Theorem 2]{Xiang-srg})
         $u/e=1/3$, where $e\mid 255$ and there exists $l>0$ such that $2^l\equiv -1\mod e$ and $1\leq u\leq e-1$.
	 It is easy to verify that for all divisors of 255 only $e=3,u=1$ satisfy the above conditions.
	 Then by \cite[Theorem 2]{Xiang-srg}, the Cayley graphs generated by the sets $D_i=\alpha^iK$ are
	 SRGs with the parameters $(256,85,24,30)$, where $K$ is the set of all non-zero cubes of $\gf_{2^8}$,
	 $\alpha$ is a primitive element of $\gf_{2^8}$ and $i\in\{0,1,2\}$.
	 Clearly, $K$ is exactly the image set of the APN function $x^3$ over $\gf_{2^8}$.
	 Therefore, the SRGs with the parameters $(256,85,24,30)$ from \cite[Theorem 2]{Xiang-srg} are isomorphic
	 to the one from our Theorem \ref{gengraph}(ii) by applying $F(x)=x^3$.
\end{enumerate}

By MAGMA, the SRGs from the other $15$ APN functions in the Appendix are not isomorphic to the known constructions,
and pairwise non-isomorphic.
We list the $15$ new $(256,85,24,30)$-SRGs in the following table.
The notation $\mathtt{G}$ denotes the SRG, $\mbox{Aut}(\mathtt{G})$ denotes the automorphism group of the
graph $\mathtt{G}$, $M$ denotes the adjacent matrix of $\mathtt{G}$, and $\mbox{Rank}(M)$ denotes the $2$-rank
of the adjacent matrix $M$.

\begin{table*}[htb]
    \label{imagec1}
    \tabcaption{New Negative Latin square type $(256,85,24,30)$-SRGs from APN functions}
    \centering
      \begin{tabular}{|c|c|c|c||c|c|c|c|} \hline
      No. &$|\mbox{Aut}(\mathtt{G})|$   & $\mbox{Rank}(M)$ & Remark
      & No.  &$|\mbox{Aut}(\mathtt{G})|$   & $\mbox{Rank}(M)$  & Remark     \\ \hline  
      $1$    &$2^9$        & $256$   & new     & $2, 6$     &$2^{11}$           & $256$    & new     \\ \hline  
      $3$    &$2^8$        & $256$   & new     & $4$        &$2^{10}$           & $256$   & new     \\ \hline  
      $5$    &$2^9$        & $256$   & new     & $6$        &$2^{11}$           & $256$   & new\\ \hline  
      $7$    &$2^{10}$     & $256$   & new     & $8$        &$2^{10}$           & $256$   & new \\ \hline  
      $9$    &$2^{9}$      & $256$   & new     & $10$       &$2^{10}$           & $256$   & new\\ \hline  
      $11$   &$2^{8}$      & $256$   & new     & $12$       &$2^{10}$           & $256$   & new \\ \hline  
      $13,14,17$    &$2^{11}\cdot 5\cdot 17$  & $256$ & Lemma \ref{cyclosrg}   & $15$        &$2^{10}$           & $256$  & new \\ \hline  
      $16$        &$2^{9}$           & $256$    & new                     & $18$        &$2^{10}$           & $256$ & new   \\ \hline  
    \end{tabular}
  \end{table*}

  Finally, due to the relationship between quadratic zero-difference balanced function of the form $F(x)=G(x^{p^t+1})$ and SRGs, we leave the following problem for interested readers.

\begin{problem}
        Let $t$ be a positive integer and $n\equiv 0\mod 2t$, construct
	quadratic zero-difference $p^t$-balanced functions $F(x)$ of the form $G(x^{p^t+1})$ over $\gf_{p^n}$, where
	the restriction of $G$ to $C_{p^t+1}$ is an injection. Particularly, find such functions when
	$n/2t$ is odd (as by Theorem \ref{gengraph} they can lead to negative Latin square type
	SRGs).
\end{problem}

\section{Concluding remarks}

In this paper, we introduced a new notion \textit{differentially $\delta$-vanishing functions} and studied its relationship with zero-difference balanced (ZDB) functions and differentially $\delta\rq{}$-uniform functions.  We showed that any zero-difference $\delta$-balanced functions are differentially $\delta$-vanishing, and any quadratic  differentially $\delta$-vanishing functions are differentially $\delta$-uniform. Note that the converse of these two inclusion relations are not true in general. We provided both the constructions of differentially $\delta$-vanishing functions and the characterization of such functions via the Walsh transform. In particular, we studied the quadratic zero-difference $p^t$-balanced functions over $\gf_{p^n}$ with the form $F(x)=G(x^{p^t+1})$, where $G$ restricts to the $d$-th power set $C_d$ is an injection and $n=2kt$. When $p=2$ and $t=1$, we explored the condition for $\alpha,\beta,\gamma$ such that $G(x)=x+\alpha\tr(\beta x+\gamma x^3)$ is an injection on $C_3$, and therefore obtaining new (up to CCZ-equivalence) APN functions. Interestingly, we showed that the image set (excluding $0$) of the aforementioned quadratic function $F(x)=G(x^{p^t+1})$ is a partial difference set in the group $(\gf_{p^n},+)$ with parameter (\ref{srgpara}). By using the recently discovered APN functions of the form $G(x^3)$ and by comparing the strongly regular graphs from their image sets with the known constructions, $15$ new $(256, 85, 24, 30)$-strongly regular graphs were obtained.

\section*{Acknowledgement}
We would like to thank Eric Chen for kindly sending us the generator matrices of the binary $[85,8]$ codes in the database of two-weight codes he maintains. We also thank the anonymous reviewers for their valuable suggestions which improve this paper. We particularly thank one reviewer for pointing out the relationship between zero-difference balanced functions and optimal constant composition codes, which leads to the results in Remark 5.

\newpage
\section*{Appendix}

The primitive polynomial of the field $\gf_{2^8}$ over $\gf_2$ is $x^8+x^4+x^3+x^2+1$.

{\tiny{
\begin{table*}[htb]
    \label{imagec1}
    \tabcaption{Quadratic APN functions on $\gf_{2^8}$ with the form $G(x^3)$ and $G|_{C_3}$ is injective}
    \centering
      \begin{tabular}{|l|c|} \hline
      No.        & Function         \\ \hline  
      $1$        & $\begin{array}{ll}
                     &w^{132}x^{192} + w^{37}x^{144} + w^{91}x^{132} + w^{188}x^{129} + w^{76}x^{96} + w^{162}x^{72} + w^{46}x^{66} + w^{252}x^{48} \\
                     &+w^{42}x^{36} + w^{81}x^{33} + w^{83}x^{24} + w^{13}x^{18} + w^{185}x^{12} + w^{163}x^{9} + w^{216}x^{6} + w^{181}x^{3} \\
                     \end{array}$                        \\ \hline  
      $2$        &$x^{144} + x^6 + x^3$                          \\ \hline  
      $3$        &$\begin{array}{ll}
                   &w^{91}x^{192} + w^{124}x^{144} + w^{214}x^{132} + w^{106}x^{129} + w^{59}x^{96} + w^{172}x^{72} + w^{138}x^{66} + w^{163}x^{48}\\
                   & + w^{58}x^{36} + w^{100}x^{33} + w^{32}x^{24} + w^{250}x^{18} + w^{45}x^{12} + w^{241}x^{6} + w^{157}x^{3}\\ \end{array}$
                    \\ \hline   
      $4$        &$w^{21}x^{144} + w^{183}x^{66} + w^{245}x^{33} + x^3$                       \\ \hline  
      $5$        &$\begin{array}{ll}
                    &w^{155}x^{192} + w^{96}x^{144} + w^{223}x^{132} + w^{77}x^{129} + w^{88}x^{96} + w^{232}x^{72} + w^{69}x^{66} + w^{142}x^{48}\\
                    &+ w^{168}x^{36} + x^{33} + w^{145}x^{24} + w^{234}x^{18} + w^{202}x^{12} + w^{94}x^{9} + w^{189}x^{6} + w^{241}x^{3}
                    \end{array}$                       \\ \hline  
      $6$        &$x^{72} + x^6 + x^3$                       \\ \hline  
      $7$        &$\begin{array}{ll}
                     &w^{126}x^{192} + w^{119}x^{144} + w^{221}x^{132} + w^{222}x^{129} + w^{79}x^{96} + w^{221}x^{72} + w^{187}x^{66}\\
                     &+ w^{148}x^{48} + w^{187}x^{36} + w^{237}x^{24} + w^{231}x^{12} + w^{119}x^{9} +
                      w^{244}x^{6} + w^{236}x^{3}\end{array}$                       \\ \hline  
      $8$        &$\begin{array}{ll}
                    &w^{25}x^{192} + w^{140}x^{144} + w^{59}x^{132} + w^{129}x^{129} + w^{42}x^{96} + w^{164}x^{72} + w^{149}x^{66} + w^{119}x^{48}\\
                    & + w^{74}x^{36} + w^{211}x^{33} + w^{9}x^{24} + w^{46}x^{18} +
                      w^{130}x^{12} + w^{185}x^{9} + w^{147}x^{6} + w^{27}x^{3}\end{array}$                       \\ \hline  
      $9$        &$\begin{array}{ll}
                   &w^{151}x^{192} + w^{13}x^{144} + w^{58}x^{132} + w^{143}x^{129} + w^{110}x^{96} + wx^{72} + w^{244}x^{66} + w^{26}x^{48} \\
                   &+ w^{180}x^{36} + w^{8}x^{33} + w^{69}x^{24} + w^{76}x^{18} +
                      w^{201}x^{12} + w^{201}x^{9} + w^{19}x^{6} + w^{107}x^{3}\end{array}$                       \\ \hline  
      $10$        &$w^{135}x^{144} + w^{120}x^{66} + w^{65}x^{18} + x^3$                       \\ \hline  
      $11$        &$\begin{array}{ll}
                     &w^{113}x^{192} + w^{56}x^{144} + w^{68}x^{132} + w^{155}x^{129} + w^{91}x^{96} + w^{78}x^{72} + w^{159}x^{66} \\
                     &+ w^{30}x^{48} + w^{194}x^{36} + w^{14}x^{33} + w^{238}x^{24} + w^{91}x^{18} +
                       w^{100}x^{12} + w^{96}x^{9} + w^{222}x^{6} + w^{178}x^{3}\end{array}$                       \\ \hline  
      $12$        &$\begin{array}{ll}
                       &w^{86}x^{192} + w^{224}x^{129} + w^{163}x^{96} + w^{102}x^{66} + w^{129}x^{48} + w^{102}x^{36} + w^{170}x^{33} + w^{14}x^{24}\\
                       & + w^{170}x^{18} + w^{101}x^{12} + w^{58}x^{6} + w^{254}x^{3}\end{array}$                       \\ \hline  
      $13$        &$x^9$                       \\ \hline  
      $14$        &$x^3$                       \\ \hline  
      $15$        &$\begin{array}{ll}
                      &w^{95}x^{192} + w^{242}x^{144} + w^{195}x^{132} + w^{98}x^{129} + w^{84}x^{96} + w^{45}x^{72} + w^{234}x^{66}\\
                      & + w^{202}x^{48} + w^{159}x^{36} + w^{58}x^{33} + w^{23}x^{24} + w^{148}x^{18} +
                        w^{230}x^{12} + w^{32}x^{9} + w^{54}x^{6} + w^{41}x^{3}\end{array}$                       \\ \hline  
      $16$        &$\begin{array}{ll}
                    &w^{189}x^{192} + w^{143}x^{144} + w^{22}x^{132} + w^{21}x^{129} + w^{133}x^{96} + w^{239}x^{72} + w^{229}x^{66} + w^{31}x^{48}\\
                    &+ w^{187}x^{36} + w^{185}x^{33} + w^{68}x^{24} + w^{236}x^{18} +
                     w^{75}x^{12} + w^{91}x^{9} + w^{97}x^{6} + w^{160}x^{3}\end{array}$                       \\ \hline  
      $17$        &$x^{57}$                       \\ \hline  
      $18$        &$w^{67}x^{192} + w^{182}x^{132} + w^{24}x^{6} + x^3$                       \\ \hline  
    \end{tabular}
  \end{table*}
}}


\begin{thebibliography}{99}

\bibitem{81graph} 
	A. E. Brouwer and W. H. Haemers,
	Structure and uniqueness of the $(81,20,1,6)$ strongly regular graph,
	Discrete Math. 106/107, 77-82, (1992).

\bibitem{databasesrg} 
	A. E. Brouwer, 
	Online database of strongly regular graphs,
	\url{http://www.win.tue.nl/~aeb/graphs/srg/srgtab.html}.

\bibitem{Xiang-srg}
	A. E. Brouwer, R. M. Wilson, Q. Xiang,
	Cyclotomy and strongly regular graphs,
	Journal of Algebraic Combinatorics, 10(1), 25-28, (1999).


\bibitem{Bracken-APN} 
	C. Bracken, E. Byrne, N. Markin and G. McGuire,
	On the Walsh spectrum of a new APN function, 
	Workshop on Cryptography and Coding,
	Lecture Notes in Computer Science, Vol. 4887, 92-98, (2007).

\bibitem{Han-ZDB}
        H. Cai, X. Zeng, T. Helleseth, X. Tang and Y. Yang, 
	A new construction of zero-difference balanced functions and its applications,
	IEEE Trans. Inf. Theory 59(8), 5008--5015, (2013).

\bibitem{cyclosrg} 
	R. Calderbank and W.M. Kantor, 
	The geometry of two-weight codes,
	Bulletin of London Mathematical Society 18, 97-122, (1986).

\bibitem{Cbook1} 
	C. Carlet, 
	Boolean Functions for Cryptography and Error Correcting Codes,
	Chapter of the monograph {\em Boolean Models and Methods in Mathematics, Computer Science, and Engineering},
	Y.~Crama and P.~Hammer eds,  Cambridge University Press, 257-397, (2010).   
	Preliminary version available at
	\url{http://www-rocq.inria.fr/codes/Claude.Carlet/pubs.html}.

\bibitem{Cbook} 
        C.~Carlet,  
	Vectorial Boolean Functions for  Cryptography, 
	Chapter of the monograph {\em Boolean Models and Methods in Mathematics, Computer Science, and Engineering},
	Y.~Crama and P.~Hammer eds,  Cambridge University Press, 398-472, (2010).   
	Preliminary version available at
	\url{http://www-rocq.inria.fr/codes/Claude.Carlet/pubs.html}.

\bibitem{carletandding} 
	C. Carlet and C. Ding, 
	Highly nonlinear mappings,
	Journal of Complexity 20(2-3), 205--244, (2004).
	

\bibitem{CDN}   
	C. Carlet, C. Ding and H. Niederreiter,
	Authentication Schemes from Highly Nonlinear Functions,
	Designs, Codes and Cryptography 40, 71--79, (2006).

\bibitem{carlet-dcc2010}
C. Carlet,
Relating three nonlinearity parameters of vectorial functions and building APN functions from bent functions,
Design, Codes and Cryptography, 59(1-3), 89--109, (2011).


\bibitem{ChabVaud} 
	F. Chabaud and S. Vaudenay,
	Links between differential and linear cryptanalysis,
	Proceedings of EUROCRYPT'94, 
	Lecture Notes in Computer Science, Vol. 950, 356--365, (1995).

\bibitem{databasetwoweightcode} 
	E. Chen, 
	Online database of two-weight codes,
	\url{http://moodle.tec.hkr.se/~chen/research/2-weight-codes/}.

\bibitem{858code} 
	E. Chen, 
	Projective binary $[85,8]$ two weight codes, private communication.

\bibitem{codeandgraph} 
	E. R. van Dam and D. Fon-Der-Flaass,
	Codes, Graphs, and Schemes From Nonlinear Functions,
	European Journal of Combinatorics (24)1, 85--98, (2000).

\bibitem{ZDB-Ding}
        C. Ding,
	Optimal constant composition codes from zero-difference balanced functions,
	IEEE Trans. Inf. Theory 54(12), 5766--5770, (2008).

\bibitem{Ding-DSS}
	C. Ding,
	Optimal and perfect systems of sets,
	Journal of Combinatorial Theory, Series A 116, 109--119, (2009).

\bibitem{Ding-Tan}          
	C. Ding and Y. Tan,
	Zero-difference balanced functions with applications, 
	Journal of Statistical Theory and Practice (6)1, 3--19, (2012).

\bibitem{Ding-ZDB}
        C. Ding, Q. Wang and M. Xiong,
	Three new families of zero-difference balanced functions with applications,
	IEEE Trans. Inf. Theory 60(4), 2407--2413, (2014).

\bibitem{dobbertin} 
	H. Dobbertin, D. Mills, E. N. Muller, A. Pott and W. Willems,
	APN functions in odd characteristic,
	Discrete Mathematics 267(1-3), 95--112, (2003).

\bibitem{edel-pott} 
	Y. Edel and A. Pott,
	A new almost perfect nonlinear function which is not quadratic,
	Advances in Mathematical Communications 3(1), 59--81, (2009).

\bibitem{edel} 
	Y. Edel,
	On quadratic APN functions and dimensional dual hyperovals,
	Designs, Codes and Cryptography, 57(1), 35--44, (2010).

\bibitem{Feng-Luo}
	K. Feng and J. Luo,
	Value distributions of exponential sums from perfect nonlinear functions and their applications,
	IEEE Trans. Inform. Theory 53(9), 3035--3041 (2007).
	
\bibitem{HK} T. Helleseth and P. V. Kumar. Sequences with low correlation. In {\em Handbook
of Coding Theory}, V. Pless and W.C. Huffman Eds. Amsterdam, The Netherlands: 
Elsevier, vol. II, pp. 1765-1854,
1998.

\bibitem{Ireland-Rosen}
	K. Ireland and M. Rosen, 
	A Classical Introduction to Modern Number Theory,
	Graduate Texts in Mathematics, Vol. 84. 
	Springer-Verlag, New York (1990).

\bibitem{gohar-pott}  
	G. M. Kyureghyan and A. Pott,
	Some theorems on planar mappings,
	Arithmetic of Finite Fields, Lecture Notes in Computer Science, Vol. 5130, 
	117--122, (2008).

\bibitem{counterexample} 
	H. Heinze, 
	Applications of Schur rings in algebraic combinatorics: graphs, partial difference
	sets and cyclotomy scheme, Ph.D. thesis, University of Oldenburg, Germany, (2001).

\bibitem{Helleseth-Kholosha} 
	T. Helleseth and A. Kholosha,
	Monomial and quadratic bent functions over the finite fields of odd characteristic, 
	IEEE Trans. on Inform. Theory, 52(5), (2006) 2018--2032.

\bibitem{kumar} 
	P.V. Kumar, R.A. Scholtz and L.R. Welch.
	Generalized bent functions and their properties, 
	Journal of Combinatorial Theory, Series A 40, 90--107, (1985).

\bibitem{Lang} 
        S. Lang, 
	Cyclotomic Fields II, Graduate Texts in Mathematics, Vol. 69,
	Springer-Verlag, New York, (1980).

\bibitem{LinSch} 
	J. H. van Lint and A. Schrijver,
	Construction of strongly regular graphs, two-weight codes and partial geometries by finite fields,
	Combinatorica 1, 63-73, (1981).

\bibitem{lidl} 
        R. Lidl and H. Niederreiter,
        Finite Fields, Encyclopedia Math. Appl. {20}, Cambridge University Press, (1983).

\bibitem{nyberg} 
	K. Nyberg, 
	Perfect nonlinear S-Boxes,
        Proceedings of EUROCRYPT'91, Lecture Notes in Computer Science, Vol. 547,
	378--386, (1991).

\bibitem{ma-pds} 
	S. L. Ma, 
	A survey of partial difference sets, 
	Design, Codes and Cryptography 4, 221--261 (1994).

\bibitem{passman}  
	D. S. Passman, 
	The Algebraic Structure of Group Rings, 
	Wiley-Interscience, New York, (1977).

\bibitem{pottbook} 
	A. Pott, 
	Finite geometry and character theory, 
	Lecture Notes in Mathematics, Vol. 1601, (1995).

\bibitem{tanas} 
	A. Pott, Y. Tan, T. Feng and S. Ling,
	Association schemes arising from bent functions, Designs, Codes and Cryptography 59(1), 
	319--331, (2011).

\bibitem{tangraph} 
	Y. Tan, A. Pott and T. Feng,
	Strongly regular graphs associated with ternary bent functions,
        Journal of Combinatorial Theory, Series A, 117(6), 668--682, (2010).

\bibitem{tan-4uni} 
	Y. Tan, L. Qu, C. Tan and C. Li,
        New families of differentially $4$-uniform permutations on $\gf_{2^{2k}}$,
        Proceedings of International conferences on Sequences and Their Applications,
	Lecture Notes in Computer Science, Vol, 7280, 25--39, (2012).

\bibitem{tan-apnlist} 
	Y. Tan, 
	New quadratic APN functions on $\gf_{2^7}$ and $\gf_{2^8}$, 
	\url{www.ytan.me}.

\bibitem{yu-APN} 
	Y. Yu, M. Wang and Y. Li,
	A matrix approach for constructing quadratic APN functions,
	Proceedings of International Workshop on Coding and Cryptography, 39--47,  (2013).

\bibitem{QiZhou-ZDB}
        Q. Wang and Y. Zhou,
	Sets of zero-difference balanced functions and their applications,
	Advances in Mathematical Communications 8(1), 83--101, (2014).

\bibitem{tan-APN} 
	G. Weng, Y. Tan and G. Gong,
        On almost perfect nonlinear functions and their related algebraic objects,
        Proceedings of International Workshop on Coding and Cryptography, 48--57,  (2013).

\bibitem{weng} 
	G. Weng, W. Qiu, Z. Wang and Q. Xiang,
	Pseudo-Paley graphs and skew Hadamard difference sets from presemifields,
	Designs, Codes and Cryptography 44, 49--62, (2007).
	
\bibitem{zha} 
	Z. Zha and X. Wang,
	Almost perfect nonlinear functions in odd characteristic,
	IEEE Transactions on Information Theory 57(7), 4826--4832, (2011).

\bibitem{ZTWY}
	Z. Zhou, X. Tang, D. Wu and Y. Yang, 
	Some new classes of zero-difference balanced functions, 
	IEEE Trans. Inf. Theory 58(1), 139--145, (2012).

\end{thebibliography}
\end{document}